\documentclass[a4paper,twocolumn,11pt,accepted=2022-03-30]{quantumarticle}
\pdfoutput=1
\usepackage[utf8]{inputenc}
\usepackage[english]{babel}
\usepackage[T1]{fontenc}
\usepackage{amsmath}
\usepackage{hyperref}

\usepackage{amssymb}
\usepackage{amsmath}
\usepackage{graphicx}
\usepackage{placeins}
\usepackage{mathrsfs}

\usepackage[permil]{overpic}

\usepackage{todonotes}
\usepackage[normalem]{ulem}
\usepackage{dcolumn}
\usepackage{hyperref}
\usepackage{rotating}
\usepackage{mathtools}
\usepackage[T1]{fontenc}
\usepackage{caption}
\usepackage{accents}
\newcommand\thickbar[1]{\accentset{\rule{.65em}{.8pt}}{#1}}

\usetikzlibrary{decorations.pathreplacing}

\makeatletter
\newcommand*{\centerfloat}{%
  \parindent \z@
  \leftskip \z@ \@plus 1fil \@minus \textwidth
  \rightskip\leftskip
  \parfillskip \z@skip}
\makeatother

\usepackage[braket]{qcircuit}

\usepackage{amsthm} 
    \newtheorem{theorem}{Theorem}[section]
    \newtheorem{lemma}[theorem]{Lemma}
    \newtheorem{corollary}[theorem]{Corollary}
    
    \theoremstyle{definition}
    \newtheorem{definition}[theorem]{Definition}

\newenvironment{claim}[1]{\par\medskip\noindent\textbf{Claim:}\space#1}{}
\newenvironment{claimproof}[1]{\par\noindent\textit{Proof of claim.}\space#1}{\hfill $\blacksquare$}

\usepackage{physics}
\usepackage{qcircuit}
\usepackage{subcaption}
\usepackage[numbers]{natbib}
\usepackage{enumitem}

\DeclareMathOperator{\Sp}{Sp}

\newcommand{\F}{\mathbb{F}}

\newcommand{\transp}{\mathsf{T}}

\newcommand{\CNOT}{\mathrm{CNOT}}
\newcommand{\CZ}{\mathrm{CZ}}
\newcommand{\SWAP}{\mathrm{SWAP}}
\newcommand{\Hgate}{H}
\newcommand{\Sgate}{S}

\begin{document}

\title{Enumerating all bilocal Clifford distillation protocols through symmetry reduction}
\author{S. Jansen}
\affiliation{Delft Institute of Applied Mathematics, Delft University of Technology, The Netherlands.}
\affiliation{Korteweg-de Vries Institute for Mathematics, University of Amsterdam, The Netherlands}
\author{K. Goodenough}
\affiliation{QuTech, Delft University of Technology, Lorentzweg 1, 2628 CJ Delft, The Netherlands}
\author{S. de Bone}
\affiliation{QuTech, Delft University of Technology, Lorentzweg 1, 2628 CJ Delft, The Netherlands}
\affiliation{QuSoft, CWI, Science Park 123, 1098 XG Amsterdam, The Netherlands}
\author{D. Gijswijt}
\affiliation{Delft Institute of Applied Mathematics, Delft University of Technology, The Netherlands.}
\author{D. Elkouss}
\affiliation{QuTech, Delft University of Technology, Lorentzweg 1, 2628 CJ Delft, The Netherlands}
\affiliation{Networked Quantum Devices Unit, Okinawa Institute of Science and Technology Graduate University, Okinawa, Japan}
\begin{abstract}
Entanglement distillation is an essential building block in quantum communication protocols. Here, we study the class of near-term implementable distillation protocols that use bilocal Clifford operations followed by a single round of communication. We introduce tools to enumerate and optimise over all protocols for up to $n=5$ (not necessarily equal) Bell-diagonal states using a commodity desktop computer. Furthermore, by exploiting the symmetries of the input states, we find all protocols for up to $n=8$ copies of a Werner state. For the latter case, we present circuits that achieve the highest fidelity with perfect operations and no decoherence. These circuits have modest depth and number of two-qubit gates. Our results are based on a correspondence between distillation protocols and double cosets of the symplectic group, and improve on previously known protocols.
\end{abstract}

\maketitle

\section{Introduction}
Entanglement is an essential resource for a host of quantum communication tasks, including but not limited to secret-key generation~\cite{bennett2020quantum}, conference key agreement~\cite{chen2005conference, ribeiro2018fully}, clock synchronisation~\cite{jozsa2000quantum}, and distributed quantum computation~\cite{grover1997quantum, cirac1999distributed, nickerson2013topological}. Due to experimental limitations, entanglement is in practice always noisy, i.e.~it has a non-unit fidelity with a target perfectly entangled state. A lower fidelity can lead to a lower rate at which one can perform certain tasks, or even yield their implementation impossible. Entanglement distillation is a set of procedures that increase the fidelity of the present entanglement by transforming  multiple copies of a lower fidelity entangled state into (usually) a smaller number of copies with higher fidelity~\cite{bennett1996purification,Bennett1996,Deutsch1996} (see~\cite{dur2007entanglement} for a review). 

Entanglement distillation has been studied in different settings. One such setting corresponds to the highly idealised scenario where one is given an asymptotic number of copies of a single state, and one can perform arbitrary local operations and classical communication (LOCC)~\cite{devetak2005distillation,buscemi2010distilling, leditzky2017useful}. Such protocols can in principle require an unbounded number of rounds of classical communication between Alice and Bob, rendering them infeasible in practice. A well-known explicit asymptotic protocol is the hashing protocol~\cite{bennett1996purification}. This protocol allows for the distillation to maximally entangled states at a finite rate, given that the given input state has a high enough input fidelity. Other scenarios include purification with the help of entanglement~\cite{vedral1999bound,ishizaka2004bound, sabat2020entanglement,sabat2020entanglementb}, on a single copy only, known as filtering~\cite{gisin1996hidden, verstraete2001local}, with environment assistance~\cite{buscemi2012general}, of higher-dimensional states~\cite{martin2003single, bombin2005entanglement} or with different classes of operations than LOCC~\cite{brandao2011one,regula2019one}. On the experimental side, entanglement distillation has been realised using photonic setups~\cite{kwiat2001experimental, pan2003experimental,yamamoto2003experimental,  walther2005quantum, chen2017experimental, ecker2021experimental} (where the distilled state is not stored in a memory afterwards), ions~\cite{reichle2006experimental} (where the distillation was performed within a single node) and NV-centres in diamond~\cite{kalb2017entanglement} (where the distilled states were heralded and stored in memories for further use).

Our goal is to find good distillation protocols with modest requirements. In particular, protocols where Alice and Bob use a small number of entangled states~\cite{Rozpdek2018,fang2019non}, and require only a single round of communication after performing their local operations~\cite{Rozpdek2018}. The above class of distillation protocols were first considered in~\cite{Rozpdek2018}, where they were called measure and exchange protocols.
The semidefinite programming bounds found by Rozpedek et al.~\cite{Rozpdek2018} allow to bound the optimal performance of measure and exchange protocols. Moreover, in some particular cases the existing protocols meet the bounds allowing to establish their optimality. 
Regarding the design of protocols, a heuristic procedure called the seesaw method allows to improve existing protocols~\cite{Rozpdek2018}. More recently, Krastanov et al.~investigated a genetic optimisation method for a subset of these protocols~\cite{krastanov2019optimized} and evaluated them including noisy operations. 

Here, complementary to previous work, we find a systematic procedure to obtain good measure and exchange protocols. 

To this end, we narrow down our investigation from general measure and exchange protocols to a practically relevant subset of protocols and states. We consider the distillation of Bell-diagonal states, where we use arbitrary noiseless \emph{bilocal Clifford circuits} and measure out all but one of the qubit pairs. The measurement results are communicated between Alice and Bob, and the protocol is deemed successful if all pairs had correlated outcomes. We call this class of protocols \emph{bilocal Clifford protocols} for short. This class of protocols includes a number of relevant protocols considered before in the literature~\cite{bennett1996purification,Deutsch1996,fujii2009entanglement,briegel1998quantum,dur2003entanglement,dur1999quantum,ruan2018adaptive,vollbrecht2005interpolation,krastanov2019optimized}. 

The restriction to bilocal Clifford protocols and Bell-diagonal states allows us to reduce the finding of all bilocal Clifford protocols to enumerating all (double) cosets $\mathcal{D}_n\backslash \Sp(2n, \F_2)/ \mathcal{K}_n$. Here, $\Sp(2n,\F_2)$ is the symplectic group over the field with two elements $\F_2$, $\mathcal{K}_n$ is the (possibly trivial) subgroup that preserves the input states and $\mathcal{D}_n$ is the distillation subgroup, which is the set of operations that leave both the success probability and fidelity invariant. One of our contributions in this work is to characterise this subgroup in terms of its generators and its order. We consider two cases for the input states - general input states (i.e.~trivial symmetry group) and the $n$-fold tensor product of Werner states. For general input states, we find all protocols for up to $n=5$ entangled pairs. For an $n$-fold tensor product of Werner states, we describe an algorithm that finds a complete set of double coset representatives. This allows us to optimise over all bilocal Clifford protocols when distilling an $n$-fold tensor product of a Werner state for $n$ up to $8$ pairs.

We find that for $n=2, 3$ copies of a Werner state, the highest fidelity out of all bilocal Clifford protocols is achieved by protocols studied before in the literature. For $n=4$ to $8$, we find increased fidelities over previously considered distillation schemes. Furthermore, we find explicit circuits achieving the highest fidelity out of all bilocal Clifford protocols, see Appendix \ref{sec:diststatistics}. These circuits have comparable depth and number of two-qubit gates as previously studied protocols, highlighting also the practical feasibility of our findings.

This paper is structured as follows. In Section ~\ref{sec:preliminaries} we describe the preliminaries and notation needed throughout the paper. Section~\ref{sec:bilocalClifford} explains bilocal Clifford distillation protocols and how the optimisation over such protocols can be rephrased as an optimisation over elements from the symplectic group $\Sp(2n,\F_2)$. In Section~\ref{sec:preservdiststat} we characterise the distillation subgroup $\mathcal{D}_n$. In Section~\ref{sec:wernerstatesection} we prove a further reduction of our search space when the state to be distilled is an $n$-fold tensor product of a Werner state. In Section~\ref{sec:results} we present our optimisation results. We end with conclusions and discussions in Section~\ref{sec:conc}.

\section{Preliminaries}
\label{sec:preliminaries}
We begin by setting some relevant notation. The field with two elements is denoted by $\F_2$. We use the notation $U_i$ to denote a single-qubit operation on qubit $i$. The single-qubit operations that we use are the Pauli gates ($I$, $X$, $Y$ and $Z$), the Hadamard gate ($H$) and the phase gate ($S$). Moreover, we denote by $\CNOT_{ij}$ a controlled-NOT operation with control qubit $i$ and target qubit $j$, by $\CZ_{ij}$ a controlled-Z operation between qubits $i$ and $j$ and by $\SWAP_{ij}$ the operation that swaps qubits $i$ and $j$.

\subsection{Pauli group and Clifford group}
\label{subsec:PauliClifford}
The Pauli matrices are defined as
\begin{equation}
    \begin{split}
    I &= \begin{bmatrix} 1 & 0 \\ 0 & 1 \end{bmatrix}, \\
    Y &= \begin{bmatrix} 0 & -i \\ i & 0 \end{bmatrix},
  \end{split}
\qquad
    \begin{split}
    X &= \begin{bmatrix} 0 & 1 \\ 1 & 0 \end{bmatrix},\\
    Z &= \begin{bmatrix} 1 & 0 \\ 0 & -1 \end{bmatrix}.
    \end{split}
    \label{eq:paulimatrices}
\end{equation}
The Pauli group with phases on $n$ qubits $\mathcal{\thickbar P}_n$ consists of all $2^n \times 2^n$ matrices of the form $\lambda P_1 \otimes \dots \otimes P_n$ with $\lambda \in \{\pm 1, \pm i\}$ and $P_i \in \{I, X, Y, Z\}$ for all $i \in \{1,\dots,n\}$, together with standard matrix multiplication. Of particular interest to us is the Pauli group without any phase factors, $\mathcal{ P}_n \cong \mathcal{\thickbar P}_n/\langle i I^{\otimes n}\rangle $. Here $\langle i I^{\otimes n}\rangle$ is the subgroup generated by $i I^{\otimes n}$. We will call this the Pauli group for short. An element of the group $\mathcal{ P}_n$ is referred to as a Pauli string (of length $n$).  The order of $\mathcal{P}_n$ equals $|\mathcal{P}_n| = 4^{n}$.

An important class of gates in quantum information theory are the so-called Clifford gates~\cite{gottesman1998theory}. Circuits composed of Clifford gates are efficiently classically simulable, yet can be used to create complex quantum states, which are used for example in stabiliser error correction. The Clifford gates on $n$ qubits form a group $\mathcal{C}_n$, and each $C \in \mathcal{C}_n$ induces an automorphism $f: \mathcal{\thickbar P}_n \rightarrow \mathcal{\thickbar P}_n$ on $\mathcal{\thickbar P}_n$ by conjugating each element with $C$, i.e.~$f(P) = CPC^{\dagger}$. The Clifford group $\mathcal{C}_n$ is generated by Hadamard- ($\Hgate_i$) and phase ($\Sgate_i$) gates on each qubit ($1\leq i \leq n$) and $\CNOT$ gates between every pair $(i, j)$ of qubits. In matrix representation, these gates are given by 

\begin{gather}
    \Hgate = \frac{1}{\sqrt{2}}\begin{bmatrix} 1 & 1 \\ 1 & -1 \end{bmatrix}, \qquad
    \Sgate = \begin{bmatrix} 1 & 0 \\ 0 & i \end{bmatrix},\\
        \qquad \CNOT = \begin{bmatrix} 1 & 0 & 0 & 0 \\ 0 & 1 & 0 & 0 \\ 0 & 0 & 0 & 1 \\ 0 & 0 & 1 & 0 \end{bmatrix}.
    \label{eq:cliffordmatrices}
\end{gather}

\subsection{Binary representation of the Pauli group and the Clifford group}
\label{subsec:binaryrep}
The elements of the Pauli group and the Clifford group can be described in terms of binary vectors and matrices, respectively. To see this, we first introduce the following notation for the Pauli matrices.
\begin{equation}
    \tau_{00} = I,\quad \tau_{10} = X,\quad \tau_{11} = iY,\quad  \tau_{01} = Z.
    \label{eq:binaryP1}
\end{equation}
\\
We extend this notation to tensor products of Pauli matrices as follows.
\begin{equation}
    \tau_v \coloneqq \tau_{v_1v_{n+1}}\otimes\dots\otimes\tau_{v_nv_{2n}},\quad v \in \mathbb{F}_2^{2n},
    \label{eq:binaryPn}
\end{equation}

As mentioned in Section \ref{subsec:PauliClifford}, the global phase factors are not important in the context of this paper, so an element $\lambda\tau_v$, $\lambda \in \{\pm 1, \pm i\}$, of $\mathcal{\thickbar P}_n$ can be represented by the binary vector $v \in \F^{2n}_2$. The multiplication of the elements of $\mathcal{\thickbar P}_n$  corresponds then to the addition of the binary vectors.

For any $C \in \mathcal{C}_n$, the conjugation map $f$ corresponds to a linear map on the set of binary vectors (and thus on $\mathcal{\thickbar P}_n$). The map $f$ is an automorphism, and thus preserves the commutation relations of the elements of $\mathcal{\thickbar P}_n$. To see what this implies for the linear transformation in the binary picture, let $v, w \in \mathbb{F}_2^{2n}$. Then
\begin{equation}
    \tau_{v}\tau_{w} = (-1)^{v^\transp \Omega w}\tau_{w}\tau_{v}, 
    \label{eq:productpaulistrings}
\end{equation}
where $\Omega = \begin{bmatrix} 0 & I_n \\ I_n & 0 \end{bmatrix}$. A proof of this formula can be found in Appendix \ref{sec:backgroundbinary}. 

Let $M$ denote the linear transformation corresponding to conjugation by $C$. It follows from equation~\eqref{eq:productpaulistrings} that
\begin{equation}
    \tau_{Mv}\tau_{Mw} = (-1)^{(Mv)^\transp \Omega Mw} \tau_{Mw}\tau_{Mv}.
    \label{eq:actionM}
\end{equation}

By equation~\eqref{eq:productpaulistrings}, we know that $\tau_v$ and $\tau_w$ commute iff $v^\transp\Omega w = 0$ and anti-commute iff $v^\transp\Omega w = 1$. In order to preserve the commutation relations, it must then hold that $v^\transp M^\transp \Omega Mw = v^\transp \Omega w$ for all $v, w \in \F_2^{2n}$, so $M^\transp \Omega M = \Omega$. The matrices $M$ that satisfy this condition thus preserve the so-called symplectic inner product $\omega(v, w) \equiv v^\transp\Omega w$ between any two $v, w \in \F_2^{2n}$. These matrices form a group known as the symplectic group over $\F_2$, denoted by $\Sp(2n, \F_2)$. The order of the symplectic group over $\F_2$ is well-known~\cite{Artin1957} to be equal to  
\begin{equation}
    |\Sp(2n, \F_2)| = 2^{n^2} \prod_{j=1}^n (4^j - 1).
    \label{eq:ordersymplectic}
\end{equation}

The symplectic complement of a subspace $V$ of $\F_2^{2n}$ is defined as the set of elements of $\F_2^{2n}$ that have zero symplectic inner product with all elements from $V$,

\begin{equation}
    V^{\perp} = \lbrace v \in \F_2^{2n}  \mid \omega(v, w) = 0~\forall w \in V \rbrace\ .
    \label{eq:sympcomplement}
\end{equation}

The symplectic complement satisfies the following property,
\begin{equation}
    \left(V^\perp\right)^\perp = V\ .
    \label{eq:sympcomplementperp}
    \end{equation}

Calculations involving a symplectic matrix $M$ can often be simplified by writing it as a block matrix $M = \begin{bmatrix} A & B \\ C & D\end{bmatrix}$, with $A, B, C, D \in M_{n \times n}(\F_2)$. From the condition $M^\transp \Omega M = \Omega$ it follows that the blocks satisfy
\begin{equation}
    \begin{split}
        & B^\transp D + D^\transp B = 0, \\ 
        & A^\transp C + C^\transp A = 0, \\ 
        & A^\transp D + C^\transp B = I_n.
    \end{split}
    \label{eq:symplecticblocks}
\end{equation}
Moreover, the inverse of $M$ is given by
\begin{equation}
    M^{-1} = \begin{bmatrix} D^\transp & B^\transp \\ C^\transp & A^\transp\end{bmatrix}.
    \label{eq:symplecticinverse}
\end{equation}

Let $\phi:\mathcal{C}_n\to \Sp(2n,\F_2)$ be the function that maps every Clifford gate to the corresponding symplectic matrix. This map is a surjective group homomorphism~\cite{Dehaene2003}. The symplectic group $\Sp(2n, \mathbb{F}_2)$ is thus generated by the images of a generating set of the Clifford group $\mathcal{C}_n$ under $\phi$.

\section{Bilocal Clifford protocols}
\label{sec:bilocalClifford}

\begin{figure}
% !TEX root = main.tex

\begin{tikzpicture}
\node[scale=1.25] at (-1.25,0.7){1)};
\node[scale=1.4] at (-1,0){$\bigotimes_{i=1}^{n}\rho_i$};
\node[scale=1.4] at (5.3,1/2){$A$};
\node[scale=1.4] at (5.3,-1/2){$B$};
\draw[line width=1.25] (0,0) -- (1,1/2);
\draw[line width=1.25] (5,1/2) -- (1,1/2);
\draw[line width=1.25] (0,0) -- (1,-1/2);
\draw[line width=1.25] (5,-1/2) -- (1,-1/2);
% \node[draw, scale=1.5, fill=white] at (3,-1/2){$P$};
% \node[draw, scale=1.5, fill=white] at (2,-1/2){$C$};
\node[draw, scale=1.25, fill=white] at (4.1,1/2){$C^{T}$};
\node[draw, scale=1.25, fill=white] at (4.1,-1/2){$C^{\dagger}$};
\node[anchor=east] at (5.5, 0) (node1){};
\node[anchor=east] at (5.5, -1.25) (node2){};
\draw[->, line width=0.3mm] (node1) to [out = -45, in = 45, looseness = 1] (node2);
\end{tikzpicture}
\vspace*{-5mm}
% !TEX root = main.tex

\begin{tikzpicture}
\node[scale=1.25] at (-1.25,0.7){2)};
\node[scale=1.4] at (-1,0){$\ket{\Phi^+}^{\otimes n}$};
\node[scale=1.4] at (5.3,1/2){$A$};
\node[scale=1.4] at (5.3,-1/2){$B$};
\draw[line width=1.25] (0,0) -- (1,1/2);
\draw[line width=1.25] (5,1/2) -- (1,1/2);
\draw[line width=1.25] (0,0) -- (1,-1/2);
\draw[line width=1.25] (5,-1/2) -- (1,-1/2);
\node[draw, scale=1.33, fill=white] at (3,-1/2){$\mathcal{N}_P$};
\node[draw, scale=1.25, fill=white] at (4.1,1/2){$C^T$};
\node[draw, scale=1.25, fill=white] at (4.1,-1/2){$C^{\dagger}$};
\node[anchor=east] at (5.5, 0) (node1){};
\node[anchor=east] at (5.5, -1.25) (node2){};
\draw[->, line width=0.3mm] (node1) to [out = -45, in = 45, looseness = 1] (node2);
\end{tikzpicture}
\vspace*{-5mm}
% !TEX root = main.tex

\begin{tikzpicture}
\node[scale=1.25] at (-1.25,0.7){3)};
\node[scale=1.4] at (-1,0){$\ket{\Phi^+}^{\otimes n}$};
\node[scale=1.4] at (5.3,1/2){$A$};
\node[scale=1.4] at (5.3,-1/2){$B$};
\draw[line width=1.25] (0,0) -- (1,1/2);
\draw[line width=1.25] (5,1/2) -- (1,1/2);
\draw[line width=1.25] (0,0) -- (1,-1/2);
\draw[line width=1.25] (5,-1/2) -- (1,-1/2);
\node[draw, scale=1.33, fill=white] at (3,-1/2){$\mathcal{N}_P$};
\node[draw, scale=1.25, fill=white] at (2,-1/2){$C$};
\node[draw, scale=1.25, fill=white] at (4.1,-1/2){$C^{\dagger}$};
\node[anchor=east] at (5.5, 0) (node1){};
\node[anchor=east] at (5.5, -1.25) (node2){};
\draw[->, line width=0.3mm] (node1) to [out = -45, in = 45, looseness = 1] (node2);
\end{tikzpicture}
\vspace*{-5mm}
% !TEX root = main.tex

\begin{tikzpicture}
\node[scale=1.25] at (-1.45,0.7){4)};
\node[scale=1.4] at (-1.25,0){$\ket{\Phi^+}^{\otimes n}$};
\node[scale=1.4] at (5.3,1/2){$A$};
\node[scale=1.4] at (5.3,-1/2){$B$};
\draw[line width=1.25] (0,0) -- (1,1/2);
\draw[line width=1.25] (5,1/2) -- (1,1/2);
\draw[line width=1.25] (0,0) -- (1,-1/2);
\draw[line width=1.25] (5,-1/2) -- (1,-1/2);
\node[draw, scale=1.33, fill=white] at (2.8,-1/2){$\mathcal{N}_{\tilde{P}}$};
\end{tikzpicture}
\caption{Schematic description of how bilocal Clifford circuits map $n$-qubit bipartite systems to $n$-qubit bipartite systems. From 1) to 2), we rewrite the state as $\otimes_{i=1}^n\rho_i = \left(I\otimes \mathcal{N}\right)\left(\ket{\Phi^+}^{\otimes n}\right)$, where $\mathcal{N}\left(\cdot \right) = \sum_{P\in \mathcal{P}_n} p_{P} P\left(\cdot\right)  P^{\dagger}$.
In 3), we use the fact that $A^\transp\otimes I\ket{\Phi^+}^{\otimes n} = I\otimes A\ket{\Phi^+}^{\otimes n}$ for any $2^n \times 2^n$ matrix $A$~\cite{wilde2011classical}. For 4), we use the fact the Cliffords act on the group of Pauli strings $\mathcal{P}_n$.}
\label{fig:pingpong}
\end{figure}
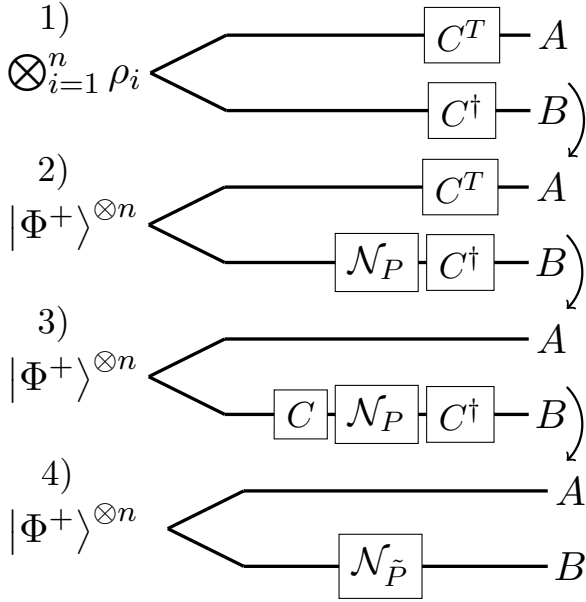

This section covers the structure of the distillation protocols that are considered in this paper. We consider a system consisting of two parties, Alice and Bob, that share $n$ entangled two-qubit states. We focus on states that are diagonal in the Bell basis. Bell-diagonal states naturally arise with realistic noise models such as dephasing and depolarizing. Moreover, any bipartite state can be \emph{twirled} into a Bell-diagonal state while preserving the fidelity~\cite{Bennett1996}. We note that the protocols found in our paper are also relevant to states that are not necessarily Bell diagonal --- the performance of a protocol on a Bell-diagonal and a not necessarily Bell-diagonal state will be comparable as long as the two states in question are close in trace distance.

Bell-diagonal states can be written as
\begin{equation}
\begin{split}
    \rho &= p_I\ket{\Phi^+}\bra{\Phi^+} + p_X\ket{\Psi^+}\bra{\Psi^+}\\ &+ p_Y\ket{\Psi^-}\bra{\Psi^-}+ p_Z\ket{\Phi^-}\bra{\Phi^-}.
\end{split}
\label{eq:belldiagonal}
\end{equation}
The indices of the probabilities arise from the following correspondence between the Bell states and the Pauli matrices.
\begin{equation}
    \begin{split}
        \ket{\Phi^+} &= \frac{1}{\sqrt{2}}\left(\ket{00}+\ket{11}\right) = \left(I \otimes I\right)\ket{\Phi^+},\\
        \ket{\Psi^+} &= \frac{1}{\sqrt{2}}\left(\ket{01}+\ket{10}\right) = \left(I \otimes X\right)\ket{\Phi^+},\\
        \ket{\Psi^-} &= \frac{1}{\sqrt{2}}\left(\ket{01}-\ket{10}\right) = \left(I \otimes -iY\right)\ket{\Phi^+},\\
        \ket{\Phi^-} &= \frac{1}{\sqrt{2}}\left(\ket{00}-\ket{11}\right) = \left(I \otimes Z\right)\ket{\Phi^+}.
    \end{split}
    \label{eq:bellstates}
\end{equation}
Equation~\eqref{eq:bellstates} gives rise to a bijective mapping from the Bell states $\ket{\Phi^+}$, $\ket{\Psi^+}$, $\ket{\Psi^-}$ and $\ket{\Phi^-}$ to the Pauli matrices $I$, $X$, $Y$ and $Z$, respectively. We denote a tensor product of $n$ Bell-diagonal states by a tensor product of Pauli matrices, e.g. $\ket{\Phi^+}\otimes\ket{\Psi^+}\otimes\ket{\Psi^-}\otimes\ket{\Phi^-}$ is denoted by $I\otimes X\otimes Y\otimes Z$.

We generalise the notation of equation \eqref{eq:belldiagonal} and denote by $p_P$ the probability that the system is in the state described by $P \in \mathcal{P}_n$. In the subscript we will not explicitly denote the tensor product, e.g. $p_{XY}$ denotes the probability that the system is described by $X \otimes Y$. The initial state of the protocol consisting of $n$ entangled two-qubit states can thus be fully described by the set of probabilities $\mathcal{Q} = \{p_{P_1P_2...P_n}: P_i \in \{I, X, Y, Z\}\}$. We refer to such a system as an \textit{$n$-qubit bipartite system}.
When working in the binary representation, we write $p_v$ instead of $p_P$, where $v$ is a binary vector and $p_v = p_P$ if $v$ is the binary representation of $P$. It will be clear from the context which convention is used.

\subsection{Bilocal Clifford circuits}
The first step of the protocol is the performance of bilocal Clifford operations. That is, if Alice applies a Clifford operation $\tilde{C} \in \mathcal{C}_n$ to her qubits, then Bob applies $\tilde{C}^*$, the entry-wise complex conjugate of $\tilde{C}$, to his qubits (see Fig.~\ref{fig:pingpong}). This leads to a permutation of the set $\mathcal{Q}$. In particular, each element $p_P$ of $\mathcal{Q}$ is mapped~\cite{wilde2011classical} to $p_{\tilde{C}^\transp P\tilde{C}^*}$, or equivalently, $p_{CPC^{\dagger}}$, where we defined $C = \tilde{C}^\transp \in \mathcal{C}_n$. We denote the probabilities that describe the permuted state by $\tilde{p}_{P_1P_2...P_n}$.

We note here that the most general permutation on $\mathcal{Q}$ by local unitaries consists of applying bilocal Cliffords followed by a Pauli string applied to either Alice or Bob's side~\cite{Dehaene2003}. These Pauli strings can be used to reorder locally the coefficients of the states.

Since (bilocal) Clifford operations form a group, the Clifford group has a group action on $\mathcal{Q}$. The (normal) subgroup of the Clifford group that fixes $\mathcal{Q}$ point-wise does not change any of the statistics, and is thus not of interest to us. This subgroup consists of all Pauli strings, and quotienting out the Cliffords by this subgroup leads to the symplectic group (over $\F_2$), $\Sp(2n, \F_2)$~\cite{Dehaene2003, dehaene2003clifford}. We can thus describe a bilocal Clifford operation by an element $M \in \Sp(2n, \mathbb{F}_2)$. To simplify notation, we sometimes slightly abuse the notation and denote by $C \in \Sp(2n,\mathbb{F}_2)$ the symplectic matrix corresponding to conjugation by $C \in \mathcal{C}_n$, but it should be kept in mind that always the symplectic matrix $M$ is meant.

\subsection{Measurements and postselection for bilocal Clifford protocols}
In the second step, Alice and Bob perform measurements in the computational basis on $n-1$ of their qubits. Alice and Bob report their results to each other using classical communication. If the outcomes are equal, they keep the state that was not measured. In this case, the protocol is called \emph{successful}. If the outcomes are not equal, they discard all states, and the protocol is not successful.
The probability that a protocol is successful is equal to the probability that all measured states are either in the $\ket{\Phi^+}$ or in the $\ket{\Phi^-}$ state, which correspond to the $I$ and $Z$ Pauli matrix, respectively. The success probability of the protocol is thus equal to
\begin{equation}
    p_\textrm{suc} = \sum_{\mathclap{\substack{P_1 \in \{I, X, Y, Z\},\\ Q_j \in \{I,Z\}}}} ~\tilde{p}_{P_1Q_2...Q_{n}},
    \label{eq:sucprobprotocol}
\end{equation}
where we used the convention that the first two-qubit state is not measured. Moreover, the fidelity between the remaining state and the $\ket{\Phi^+}$ state is equal to 
\begin{equation}
    F_\textrm{out} = \frac{\sum_{Q_j \in \{I,Z\}} \tilde{p}_{I_1Q_2...Q_{n}}}{p_\textrm{suc}}.
    \label{eq:fidelityprotocol}
\end{equation}

To simplify notation in the rest of this paper, we introduce the following two definitions.
\begin{definition}
    The \textit{base} of an $n$-qubit bipartite quantum system is given by 
    \begin{equation*}
        \begin{split}
            \mathscr{B} = \left\{v \in \mathbb{F}_2^{2n} \mid v_i = 0\ \forall i \in \{1,...,n+1\} \right\}.
        \end{split}
    \end{equation*}
    \label{def:base}
\end{definition}

Note that the base vectors correspond to the Pauli strings $I_1\otimes Q_2 \otimes... \otimes Q_{n} \in \mathcal{P}_n$ with $Q_j \in \{I,Z\}$ for all $j \in \{2,...,n\}$.

\begin{definition}
    The \textit{pillars} of an $n$-qubit bipartite quantum system are given by 
    \begin{equation*}
    \begin{split}
        \mathscr{P} = \left\{v \in \mathbb{F}_2^{2n} \mid v_i = 0\ \forall i \in \{2,...,n\}\right\}.
    \end{split}
    \end{equation*}
    \label{def:pillars}
\end{definition}
The elements of the pillars correspond to the Pauli strings $P_1 \otimes Q_2 \otimes...\otimes Q_{n} \in \mathcal{P}_n$ with $P_1 \in \{I, X, Y, Z\}$ and $Q_j \in \{I,Z\}$ for all $j \in \{2,...,n\}$. The naming of the base and pillars is made clear when the probabilities $p_P$ are ordered in an $n$-dimensional hypercube, where each dimension corresponds to a qubit pair, see Fig.~\ref{fig:baseandpillars}.

\begin{figure*}
    \centerfloat
    \begin{subfigure}{.9\textwidth}
        \centering
        \includegraphics[clip,  width=0.45\textwidth, trim = 3.0mm 3.0mm 3.1mm 0mm]{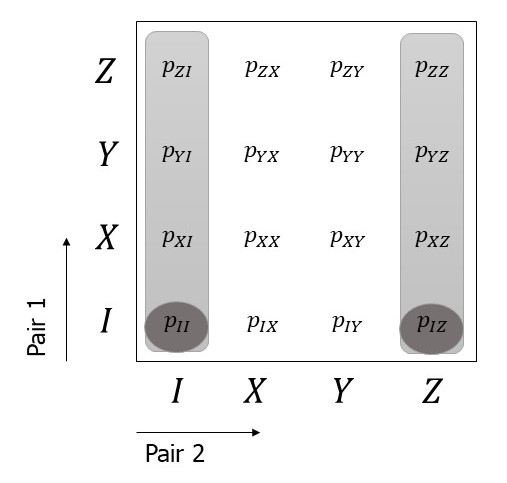}
        \caption*{}
        \label{fig:2qubitpairs}
    \end{subfigure}%
    \hskip -7cm
    \begin{subfigure}{.9\textwidth}
        \centerfloat
        \includegraphics[clip,  width=0.45\textwidth, trim = 3.0mm 3.0mm 3.1mm 0mm]{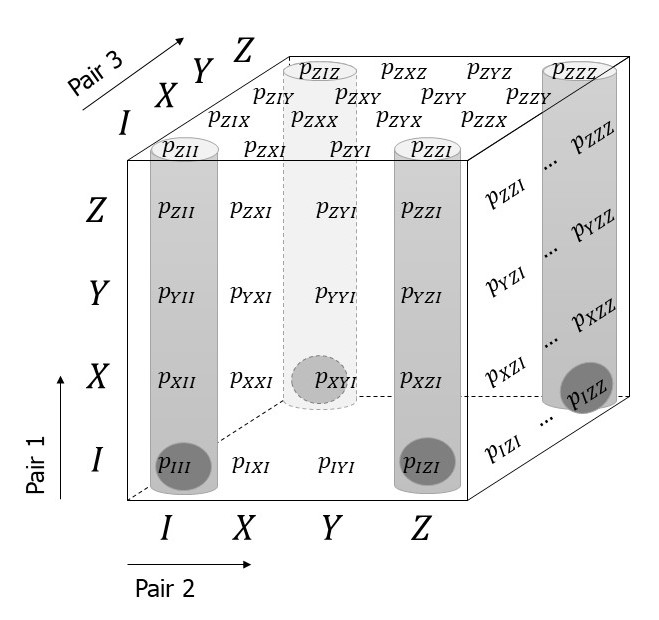}
        \caption*{}
        \label{fig:3qubitpairs}
    \end{subfigure}
    \vspace*{-6mm}
\caption{Probabilities that describe the state of a 2-pair system (left) and a 3-pair system (right). The light grey rectangles/cylinders highlight the probabilities that correspond to the pillars. The darker circles highlight the probabilities that correspond to the base. For the 3-pair system we have labelled here only the coefficients that are on the front, right and top face.}
\label{fig:baseandpillars}
\end{figure*}

Using these definitions, equation~\eqref{eq:sucprobprotocol} can be rewritten as
\begin{equation}
    p_\textrm{suc}= \sum_{v \in \mathscr{P}} \tilde{p}_v,
    \label{eq:sucprob}
\end{equation}
and equation~\eqref{eq:fidelityprotocol} as
\begin{equation}
    F = \frac{\sum_{v \in \mathscr{B}}\tilde{p}_v}{\sum_{v \in \mathscr{P}}\tilde{p}_v}.
    \label{eq:fidelity}
\end{equation}

The fidelity as defined above corresponds to the $I$ coefficient of the output state, and does not suffice to describe the output state completely. To describe the output state, we require the $X$, $Y$ and $Z$ coefficients as well. Similarly to equation~\eqref{eq:fidelity}, these coefficients $F_i$ are described in terms of the probabilities by
\begin{equation}
    F_i = \frac{\sum_{v \in \mathscr{B}+v_i}\tilde{p}_{v}}{\sum_{v \in \mathscr{P}}\tilde{p}_v}.
    \label{eq:othercoeff}
\end{equation}
Here, $v_i$ is $v_1=e_1, v_2=e_1+e_{n+1}, v_3=e_{n+1}$, corresponding to the $X$, $Y$ and $Z$ coefficients, respectively and we have used the standard basis vectors $\{e_i: i \in \{1,...,2n\}\}$ of $\mathbb{F}_2^{2n}$.

The fidelity, success probability and the three $F_i$ coefficients are referred to as the \textit{distillation statistics}. Importantly, we are not interested in permutations of the three $F_i$ coefficients, since they can be permuted arbitrarily by local operations after the measurement step.

In the binary picture, the distillation statistics can be calculated using the inverse of the symplectic matrix, which can be efficiently calculated using~\eqref{eq:symplecticinverse}. Let $M$ be the symplectic matrix corresponding to a permutation $P \mapsto CPC^\dagger$, $C \in \mathcal{C}_n$. We wish to determine which binary vectors are mapped to the vectors corresponding to the base and the pillars by $M$. Since $M$ permutes the binary vectors, this is equivalent to determining where the base and pillar vectors are mapped to by $M^{-1}$.

Finally, there is a direct analogy between our optimisation over bilocal Clifford protocols, and quantum error detection schemes of the form shown in Fig.~\ref{fig:errordetection}. Such schemes will detect as errors the set of Pauli strings that do not end up in the pillars after applying the Clifford circuit $C$. We will not pursue this further in this paper, however.

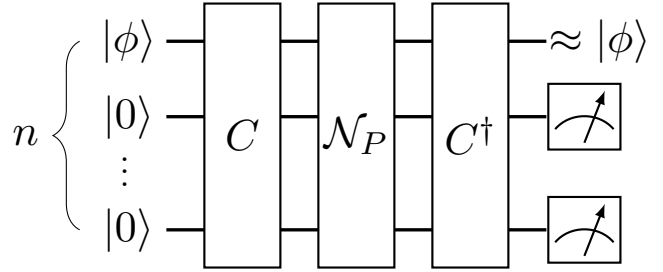
\begin{figure}
\centerfloat
% !TEX root = main.tex

\tikzset{meter/.append style={draw, inner sep=10, rectangle, font=\vphantom{A}, minimum width=30, line width=.8,
 path picture={\draw[black] ([shift={(.1,.3)}]path picture bounding box.south west) to[bend left=50] ([shift={(-.1,.3)}]path picture bounding box.south east);\draw[black,-latex] ([shift={(0,.1)}]path picture bounding box.south) -- ([shift={(.3,-.1)}]path picture bounding box.north);}}}

\begin{tikzpicture}
% \node[scale=1.25] at (-1.25,0.7){4)};
\node[scale=1.5] at (0+0.5,-1/2){$\ket{\phi}$};
\node[scale=1.5] at (6.7,-1/2){$\approx\ket{\phi}$};
\node[scale=1.5] at (0+0.5,-3/2){$\ket{0}$};
\node[scale=1.5] at (0+0.5,-6/2){$\ket{0}$};
\draw [decorate,decoration={brace,amplitude=10pt},xshift=-4pt,yshift=0pt]
(-0.0,-6/2) -- (-0.0,-1/2) node [black,midway,xshift=0.6cm]{};
\node[scale=1.5] at (-0.85,-1.75){$n$};
% % \node[scale=1.5] at (5.3,1/2){$A$};
% % \node[scale=1.5] at (5.3,-1/2){$B$};
% % \draw[line width=1.25] (0,0) -- (1,1/2);
% % \draw[line width=1.25] (5,1/2) -- (1,1/2);
% % \draw[line width=1.25] (0,0) -- (1,-1/2);
\draw[line width=1.25] (6,-1/2) -- (1,-1/2);
\draw[line width=1.25] (6,-3/2) -- (1,-3/2);
% \draw[line width=1.25] (5,-3/2) -- (1,-3/2);
\node[scale=1.2] at (0.45,-2.1){$\vdots$};
\draw[line width=1.25] (6,-6/2) -- (1,-6/2);
\draw [fill=white,thick] (2-0.5,-6/2-0.5) rectangle (3-0.5,0);
\draw [fill=white,thick] (3.5-0.5,-6/2-0.5) rectangle (4.5-0.5,0);
\draw [fill=white,thick] (5-0.5,-6/2-0.5) rectangle (6-0.5,0);
\node[scale=1.5, fill=white] at (2.5-0.5,-3/2-0.25){$C$};
\node[scale=1.5] at (4-0.5,-3/2-0.25){$\mathcal{N}_{P}$};
\node[scale=1.5] at (5.5-0.5,-3/2-0.25){$C^\dagger$};
% \node[draw, scale=1.5, fill=white] at (3,-1/2){$\mathcal{N}_{\tilde{P}}$};
% \node[draw, scale=1.5, fill=white] at (4.2,-1/2){$C^\dagger$};
\node[meter, scale = 0.9] (meter1) at (6.5,-3/2) {};
\node[meter, scale = 0.9] (meter1) at (6.5,-6/2) {};
\end{tikzpicture}
\caption{Equivalence between bilocal Clifford protocols and a subset of error detection schemes. This circuit detects as errors the set of Pauli strings that do not end up in the pillars after applying the Clifford circuit $C$.}
\label{fig:errordetection}
\end{figure}

\section{Preservation of distillation statistics}\label{sec:preservdiststat}
The distillation statistics from equations \eqref{eq:sucprob}, \eqref{eq:fidelity} and \eqref{eq:othercoeff} are the relevant parameters for quantifying an entanglement distillation protocol. Furthermore, there exist non-identical bilocal Clifford circuits which result in the same distillation statistics. To find all bilocal Clifford protocols, it is thus sufficient to find a representative bilocal Clifford protocol for each unique collection of distillation statistics. In this section we characterise these representatives for general input states.

First, we specify the set of bilocal Clifford operations that preserve the distillation statistics. We denote this set by $\mathcal{D}_n$. Now observe that $\mathcal{D}_n$ is a subgroup of $\Sp(2n, \mathbb{F}_2)$. Moreover, let $M \in \Sp(2n, \mathbb{F}_2)$ and consider the corresponding distillation protocol. We can freely add or remove elements from $\mathcal{D}_n$ at the end of this protocol, without changing the distillation statistics. That is, all elements in the right coset $\mathcal{D}_n M = \{DM: D \in \mathcal{D}_n\}$ yield the same distillation statistics. Instead of optimising over all possible Clifford circuits it thus suffices to optimise over the right cosets of $\mathcal{D}_n$ in $\Sp(2n, \mathbb{F}_2)$.

\subsection{Characterising the subgroup that preserves distillation statistics using base and pillars}
In this section we explain the relation between the base and the pillars, which were introduced in definitions \ref{def:base} and \ref{def:pillars}, respectively, and the distillation subgroup~$\mathcal{D}_n$.

From equations \eqref{eq:sucprob},  \eqref{eq:fidelity} and \eqref{eq:othercoeff} it can be observed that for a general initial state, the operations that preserve the distillation statistics are precisely those operations that leave simultaneously both the base and pillars invariant, and permute the three affine subspaces $\mathscr{B}+e_1, \mathscr{B}+e_{n+1}$ and $\mathscr{B}+e_1 +e_{n+1}$. In the following lemma it is first proven that invariance of the base implies invariance of the pillars, and vice versa.
\begin{lemma}
    Let $\mathcal{Q}$ be an $n$-qubit bipartite quantum system with base $\mathscr{B} \subseteq \mathbb{F}_2^{2n}$ and pillars $\mathscr{P} \subseteq \mathbb{F}_2^{2n}$. Let $\pi: \mathbb{F}_2^{2n} \rightarrow \mathbb{F}_2^{2n}$, $\pi(v) = Mv$, with $M \in \Sp(2n, \mathbb{F}_2)$. Then $\pi[\mathscr{B}] = \mathscr{B} \iff \pi[\mathscr{P}] = \mathscr{P}$.
    \label{thm:basepillarsrelation}
\end{lemma}

\begin{proof}
    We first prove $\pi[\mathscr{B}] = \mathscr{B} \implies \pi[\mathscr{P}] = \mathscr{P}$. For this, we first show that the pillars form the symplectic complement of the base, i.e.~$\mathscr{B}^\perp = \mathscr{P}$ (see equation~\eqref{eq:sympcomplement}). Recall from Definition \ref{def:base} that $v \in \mathscr{B}$ if and only if $v_i = 0$ for ${i = 1,...,n+1}$. Note that $\mathscr{B}$ is a subspace of $\F_2^{2n}$. The symplectic inner product between $v$ and $w \in \mathbb{F}_2^{2n}$, is equal to $\omega(v,w) = v^T \Omega w$. This is equal to zero for all $v \in \mathscr{B}$ if and only if $w_i = 0$ for all $i \in \{2,...,n\}$, so iff $w \in \mathscr{P}$.
    
     Let $v \in \mathscr{B}$ and $w \in \mathscr{P}$. Then $\omega(v, w) = 0$, and since $M \in \Sp(2n, \F_2^{2n})$, we have that $\omega\left(\pi(v), \pi(w)\right) = 0$ as well. Since by assumption $\pi(v) \in \mathscr{B}$, it follows that $\pi(w) \in \mathscr{P}$. Finally, since $\pi$ is an automorphism, we know that it is bijective and thus $\pi[\mathscr{P}] = \mathscr{P}$.
        
    For the other direction, we use the fact that $\mathscr{P}^\perp = \mathscr{B}$, see equation~\eqref{eq:sympcomplementperp}. Then, the above argument can be repeated with $\mathscr{B}$ and $\mathscr{P}$ interchanged to conclude that $\pi[\mathscr{B}] = \mathscr{B} \iff \pi[\mathscr{P}] =\mathscr{P}$.
\end{proof}

We will now show that not only does preservation of the base imply preservation of the fidelity and success probability, but that it also implies that the other coefficients of the output state are preserved, up to a permutation.
\begin{lemma}
Let $\pi: \mathbb{F}_2^{2n} \rightarrow \mathbb{F}_2^{2n}$, $\pi(v) = Mv$, with $M \in \Sp(2n, \mathbb{F}_2)$ such that $\pi[\mathscr{B}] = \mathscr{B}$. Then the coefficients $F_1, F_2, F_3$ are permuted amongst each other after applying $\pi$. In other words, the three coefficients of the output state are stabilised as a set after application of $M$.
    \label{thm:affinesubspaces}
\end{lemma}
\begin{proof}
We have the following decomposition of the linear space $\mathscr{P}$ into four cosets of the subspace $\mathscr{B}$:
$$
\mathscr{P}=\mathscr{B}\cup \left(\mathscr{B}+e_1\right)\cup\left( \mathscr{B}+e_1+e_{n+1}\right)\cup \left(\mathscr{B} +e_{n+1}\right).$$

Since $\pi\left[\mathscr{B}\right]=\mathscr{B}$, it follows from Lemma~\ref{thm:basepillarsrelation} that $\pi\left[\mathscr{P}\right]=\mathscr{P}$. So $\pi$ permutes the three cosets $v_i+\mathscr{B}$, $v_i=e_1,e_1+e_{n+1},e_{n+1}$.

It follows by equation~\eqref{eq:othercoeff} that $\pi$ permutes the coefficients $F_1$, $F_2$, and $F_3$.
\end{proof}

From Lemma~\ref{thm:basepillarsrelation} and \ref{thm:affinesubspaces} we conclude that the operations that preserve the distillation statistics for arbitrary input states are precisely the operations that leave the base invariant. We use this observation to characterise the subgroup $\mathcal{D}_n$ that preserves the distillation statistics. In the trivial case that $n = 1$, we have $\mathcal{D}_1 = \Sp(2,\mathbb{F}_2)$. In this case, the only base element is the identity $I$, which is always mapped to itself under an automorphism. For all $n > 1$, however, $\mathcal{D}_n$ is a proper subgroup of $\Sp(2n,\mathbb{F}_2)$. Consider for instance the Hadamard gate on the second qubit, which is an element of $\Sp(2n,\mathbb{F}_2)$. This gate induces the swap of $X_2$ and $Z_2$ and hereby changes the base. 

\subsection{Generators of the subgroup that preserves distillation statistics}
The goal of this section is to characterise the distillation subgroup $\mathcal{D}_n$. In particular, we find the distillation subgroup in terms of a generating set $T_n$.

\begin{lemma}\label{thm:distsubgroupgen}
The distillation subgroup is generated by the set $T_n$, i.e.~$\mathcal{D}_n = \langle T_n \rangle$, where

\begin{gather}
T_n = \left\lbrace \Hgate_1, \Sgate_1, \ldots, \Sgate_n\right\rbrace\, \cup\,\left\lbrace \CNOT_{ij} \mid 1\leq j < i \leq n \right \rbrace \nonumber \\
\,\cup \, \left\lbrace \CNOT_{ij} \mid 2 \leq i < j \leq n\right\rbrace.
\end{gather}
\end{lemma}
\begin{proof}
By inspection, each element of $T_n$ preserves the base, so by Lemmas~\ref{thm:basepillarsrelation} and \ref{thm:affinesubspaces} we have that $\langle T_n \rangle \subseteq \mathcal{D}_n$. For the other inclusion, let $M = \begin{bmatrix}A & B\\ C & D\end{bmatrix} \in \mathcal{D}_n$. We show that we can reduce such an arbitrary $M$ to the identity matrix by left-multiplication by elements in $\langle T_n\rangle $. An overview of the basic matrix operations corresponding to multiplication by elements of $\Sp(2n,\mathbb{F}_2^{2n})$ is given in Appendix \ref{sec:matrixoperations}.

First, note that if $M\in \mathcal{D}_n$, then by definition $M[\mathscr{B}] \subseteq \mathscr{B}$ and $M[\mathscr{P}] \subseteq \mathscr{P}$. In the binary picture this implies that $$B_{ij} = 0\textrm{ if }(i, j) \neq (1, 1), \qquad D_{12} = \cdots = D_{1n} = 0.$$

Since $M$ has full rank, we cannot have $B_{11} = D_{11} = 0$. Hence, by multiplying $M$ from the left by $I, \Hgate_1$ or $\Hgate_1\Sgate_1$, we may assume that $B_{11}=0$ (such that $B=0$) and $D_{11} = 1$.

That $M$ has full rank implies that the last $n$ columns of $M$ are linearly independent. By using $\CNOT$ gates from $\left\lbrace \CNOT_{ij} \mid 1\leq j < i \leq n \right \rbrace \subseteq T_n$ and $\left\lbrace \CNOT_{ij} \mid 2 \leq i < j \leq n\right\rbrace \subseteq T_n$, the $D$ submatrix can be reduced to the identity matrix.

Since $D = I$ and $B = 0$, it follows from~\eqref{eq:symplecticblocks} that $A = I$ and $C$ is symmetric. For $1\leq j < i$ denote $\Sgate_{ij}:=(\Sgate_j\CNOT_{ij})^2\in \langle T_n\rangle$. Left-multiplication by $S_{ij}$ corresponds to adding row $i$ to row $n+j$ and adding rows $i$ and $j$ to row $n+i$. Note that this preserves the fact that $A = D = I$ and $B=0$.

For $j = 1, \ldots, n-1$ (in this order), we can multiply $M$ from the left by elements from $\lbrace S_j\rbrace \cup \lbrace \Sgate_{ij} \mid i > j\rbrace \subseteq \langle T_n\rangle$ to ensure that $C_{ji} = 0$ for all $1\leq j \leq i \leq n$. This implies that $C$ is strictly lower triangular. But if $C$ is strictly lower triangular and symmetric, $C = 0$. This implies that $M = I$.
\end{proof}

\subsection{Order of the subgroup that preserves distillation statistics}
As noted before, for general input states it is sufficient to only consider the right cosets of $\mathcal{D}_n$ in $\Sp(2n,\mathbb{F}_2)$. To see how much looking at cosets of $\mathcal{D}_n$ in $\Sp(2n,\mathbb{F}_2)$ limits the search space of protocols, in this section a formula for the order of $\mathcal{D}_n$ is presented and proved. As mentioned earlier, in the trivial case that $n=1$ we have $\mathcal{D}_1 = \Sp(2,\mathbb{F}_2)$, and thus $|\mathcal{D}_1| = |\Sp(2,\mathbb{F}_2)| = 6$. For $n \geq 2$ the order of $\mathcal{D}_n$ is given in Theorem~\ref{thm:orderDn}.

\begin{theorem}
    For an $n$-to-$1$ distillation protocol, with $n > 1$, the order of $\mathcal{D}_n$ is given by
    $$|\mathcal{D}_n| = 6 \cdot 2^{n^2-1} \prod_{j=1}^{n-1} (2^j - 1).$$
    \label{thm:orderDn}
\end{theorem}

\begin{proof}
    First note that $D \in \mathcal{D}_n$ is fully determined by how it maps each of the standard basis vectors $\{e_i: i \in \{1,...,2n\}\}$ of $\mathbb{F}_2^{2n}$. We count how many transformations of the standard basis vectors are possible.
    
    Let us start by looking at $e_{2n}$. This is a base element, thus it must again be transformed to a base element, because $D$ preserves the distillation statistics. There are $2^{n-1}$ base elements, but the identity element, the zero vector, is always mapped to itself by $D$. Thus there are $2^{n-1} - 1$ possibilities for the transformation of $e_{2n}$. That all transformations are indeed possible, is proved by giving a construction. Suppose that $e_{2n}$ is mapped to a base element $v \equiv De_{2n} \in \mathscr{B}$. We show that $v$ can be transformed to $e_{2n}$ through left-multiplication by elements of $\mathcal{D}_n$, see Appendix \ref{sec:matrixoperations}. The transformation from $e_{2n}$ to $v$ can then be obtained by taking the product of the inverses of these generators in reverse order.
    
    Note that $v_1,...,v_{n+1} = 0$. The vector $v$ can be transformed to $e_{2n}$ by taking the following steps. 
    \begin{enumerate}
        \item If $v_{2n} = 0$, apply a $\CNOT_{ni} \in \mathcal{D}_n$ gate with $i$ chosen such that $v_{n+i} = 1$. Note that there always is a $i$ such that this is possible, because otherwise $v$ is the zero vector, which corresponds to the identity element $I^{\otimes n}$. 
        \item For all $i \in \{2,...,n\}$ with $v_{n+i} = 1$, apply a $\CNOT_{in}\in \mathcal{D}_n$ gate.
    \end{enumerate}
    Steps 1 and 2 are visually summarised below.
    \begin{equation*}
        v =
        \left[ \begin{array}{c}
        0 \\ 0 \\ 0 \\ \hline
        0 \\ \cdot \\ \cdot
        \end{array} \right]
        \xrightarrow{1} 
        \left[ \begin{array}{c}
        0 \\ 0 \\ 0 \\ \hline
        0 \\ \cdot \\ 1
        \end{array} \right]
        \xrightarrow{2} 
        \left[ \begin{array}{c}
        0 \\ 0 \\ 0 \\ \hline
        0 \\ 0 \\ 1
        \end{array} \right] = e_{2n}
    \end{equation*}
    
    Given the transformation of $e_{2n}$ by $D$, we now wish to determine the number of possible transformations for $e_n$. We know that left-multiplication by $D$ preserves the symplectic inner product. Hence, since $\omega(e_n, e_{2n}) = 1$, it must hold that $\omega(De_n, De_{2n}) = 1$. Observe that for every non-zero element $u \in \mathbb{F}_2^{2n}$, exactly for half of the elements of $\mathbb{F}_2^{2n}$ the symplectic inner product with $u$ is equal to one\footnote{Let $u \in \mathbb{F}_2^{2n}$, such that $u_k = 1$. Then the vectors $u' \in \mathbb{F}_2^{2n}$ satisfying $\omega(u, u') = 1$ can be constructed by choosing $u'_j \in \{0,1\}$ randomly for $j \in \{1,\dots,2n\}\setminus \{n+k\}$ and then choosing $u'_{n+k} \in \{0,1\}$ such that $\omega(u, u')= 1$.}. Thus there are $\frac{|\mathbb{F}_2^{2n}|}{2} = \frac{4^n}{2} = 2^{2n-1}$ possibilities for the transformation of $e_n$.
    
    We show that each of those transformations can indeed be achieved. Suppose that $D$ has mapped $e_n$ to a vector $w \equiv De_n \in \mathbb{F}_2^{2n}$. Because $\omega(De_n, De_{2n}) = 1$ and $De_{2n}$ is a base vector, we know that there is at least one ${i \in \{2,\dots,n\}}$ such that $w_i = 1$. Since we can always apply a $\CNOT_{in}\in \mathcal{D}_n$ gate, which does not affect the vector $e_{2n}$, we can assume without loss of generality that $w_n = 1$. Now $w$ can be transformed to $e_n$ without affecting $e_{2n}$ by taking the following steps.
    \begin{enumerate}[resume]
        \item For all $i$ with $w_i = 1$ apply a $\CNOT_{ni}\in \mathcal{D}_n$ gate. 
        \item For all $i\neq n$ with $w_{n+i} = 1$, apply a $\CZ_{in}\in \mathcal{D}_n$ gate. This operation results in the addition of row $i$ to row $2n$ and the addition of row $n$ to row $n+i$. Note that this operation leaves the base invariant, so indeed $\CZ_{in} \in \mathcal{D}_n$.
        \item If $w_{2n} = 1$, apply the gate $S_n \in \mathcal{D}_n$ on qubit $n$.
    \end{enumerate}
    Steps 3 to 5 are visually summarised below.
    \begin{equation*}
        w =
        \left[ \begin{array}{c}
        \cdot \\ \cdot \\ 1 \\ \hline
        \cdot \\ \cdot \\ \cdot
        \end{array} \right]
        \xrightarrow{3} 
        \left[ \begin{array}{c}
        0 \\ 0 \\ 1 \\ \hline
        \cdot \\ \cdot \\ \cdot
        \end{array} \right]
        \xrightarrow{4} 
        \left[ \begin{array}{c}
        0 \\ 0 \\ 1 \\ \hline
        0 \\ 0 \\ \cdot
        \end{array} \right]
        \xrightarrow{5} 
        \left[ \begin{array}{c}
        0 \\ 0 \\ 1 \\ \hline
        0 \\ 0 \\ 0
        \end{array} \right]= e_n
    \end{equation*}
    Thus indeed, given the transformation of $e_{2n}$, there are $2^{2n-1}$ possible transformations of $e_n$. Combining this with the number of transformations of $e_{2n}$, we find that there are $2^{2n-1}(2^{n-1}-1)$ possible transformations for $e_n$ and $e_{2n}$ together.
    
    The elements of $\mathcal{D}_n$ that leave $e_n$ and $e_{2n}$ invariant form a subgroup that is isomorphic to $\mathcal{D}_{n-1}$, with the number of cosets in $\mathcal{D}_n$ equal to $2^{2n-1}(2^{n-1} - 1)$. Thus
    $$|\mathcal{D}_n| = 2^{2n-1}(2^{n-1}-1)|\mathcal{D}_{n-1}|.$$
    By induction on $n$ it follows that \begin{equation*}
    \begin{split}
        |\mathcal{D}_n| &= |\mathcal{D}_1| \prod_{j=2}^n 2^{2j-1}(2^{j-1} - 1) \\
        &= 6 \cdot 2^{\sum_{j=2}^n (2j-1)} \prod_{j=2}^{n} (2^{j-1} -1)\\
        &= 6 \cdot 2^{n^2 - 1} \prod_{j=1}^{n-1}(2^j - 1).
    \end{split}
    \end{equation*}
\end{proof}
The following corollary is a direct consequence of Theorem \ref{thm:orderDn}.
\begin{corollary}
    The index of $\mathcal{D}_n$ in $\Sp(2n,\mathbb{F}_2)$ is given by
    $$ \left[\Sp(2n,\mathbb{F}_2) : \mathcal{D}_n\right] = \frac{1}{3} (2^n - 1) \prod_{j=1}^n (2^j + 1).$$
\end{corollary}

\begin{proof}
    Recall that $|\Sp(2n,\mathbb{F}_2)| = 2^{n^2} \prod_{j=1}^n (4^j -1)$. As a result,
    \begin{equation*}
        \begin{split}
            \left[\Sp(2n,\mathbb{F}_2) : \mathcal{D}_n\right] &= \frac{|\Sp(2n,\mathbb{F}_2)|}{|\mathcal{D}_n|}\\
            &= \frac{2^{n^2} \prod_{j=1}^n (4^j -1)}{6 \cdot 2^{n^2 - 1} \prod_{j=1}^{n-1}(2^j - 1)}\\
            &= \frac{\prod_{j=1}^n (2^j -1)(2^j + 1)}{3 \prod_{j=1}^{n-1}(2^j - 1)}\\
            &= \frac{1}{3} (2^n -1) \prod_{j=1}^n(2^j + 1).
        \end{split}
    \end{equation*}
\end{proof}

For comparison, we list the values of $|\Sp(2n,\mathbb{F}_2)|$ and $[\Sp(2n,\mathbb{F}_2):\mathcal{D}_n]$ in Table \ref{table:orderindex} for $n = 2, 3, 4, 5$.

\begin{table*}[t]
\centerfloat
\begin{tabular}{l || c | c | c | c}
 & 2 & 3 & 4 & 5 \\
 &&&&\\[-0.8em]
\hline
 &&&&\\[-0.8em]
$|\Sp(2n,\mathbb{F}_2)|$ & 720 & 1451520 & 47377612800 & 24815256521932800 \\
 &&&&\\[-0.8em]
 \hline
 &&&&\\[-0.8em]
$[\Sp(2n,\mathbb{F}_2):\mathcal{D}_n]$ & 15 & 315 & 11475 & 782595
\end{tabular}
\caption{Values of $|\Sp(2n,\mathbb{F}_2)|$ and $[\Sp(2n,\mathbb{F}_2):\mathcal{D}_n]$ for $n = 2, 3, 4, 5$.}
\label{table:orderindex}
\end{table*}

\subsection{Finding a transversal}
In this section, we briefly describe a way to find a transversal for the cosets of $\mathcal{D}_n$ in $\Sp(2n,\mathbb{F}_2)$. A transversal is a set that contains exactly one element for each of the cosets. Once this transversal is found, it can be applied to an arbitrary $n$-qubit input state to calculate all possible distillation statistics that can be achieved using bilocal Clifford circuits. From this set of distillation statistics, the optimal protocol based on any optimality criterion can be selected.

In order to find a transversal, random elements from the symplectic group $\Sp(2n, \F_2)$ are sampled. A sampled element is added to the set of representatives if the corresponding coset is not yet represented in this set. Recall that two elements belong to the same coset if they result in the same distillation statistics (for a general input state). This is the case if and only if the same Pauli strings are mapped to the base. More formally, consider an $n$-qubit pairs bipartite system with base $\mathscr{B}$ in the binary picture. Let $M_1$, $M_2 \in \Sp(2n, \mathbb{F}_2)$. Let $\mathscr{V}$ denote the set of binary vectors that are mapped to the base by $M_1$ and let $\mathscr{W}$ denote the set of binary vectors that are mapped to the base by $M_2$. Then $M_1$ and $M_2$ belong to the same coset if and only if $\mathscr{V} = \mathscr{W}$. Because $M_1$ and $M_2$ permute the binary Pauli vectors, this is equivalent to $M_1^{-1}[\mathscr{B}] = M_2^{-1}[\mathscr{B}]$.

The sampling is continued until the set of representatives has size $[\Sp(2n,\mathbb{F}_2):\mathcal{D}_n]$. Note that finding a transversal in the way described in this section is equivalent to the coupon collector's problem. Hence, it has expected running time ${\mathcal{O}([\Sp(2n,\mathbb{F}_2):\mathcal{D}_n]\log[\Sp(2n,\mathbb{F}_2):\mathcal{D}_n])}$.

\section{Reduction for $n$-fold tensor products of Werner states}\label{sec:wernerstatesection}
Here we describe our reduction of the search space when the input state is an $n$-fold tensor product of Werner states. A Werner state has coefficients $p_I = F_\textrm{in},~p_X = p_Y = p_Z = \frac{1-F_\textrm{in}}{3}$, and its $n$-fold tensor product is highly symmetric --- it is left invariant under any element of $\mathcal{K}_n$, where $$\mathcal{K}_n=\langle\, \{\SWAP_{ij}\}_{1\leq i<j\leq n}\,\cup\,\{\Hgate_i\}_{i=1}^n\,\cup\,\{\Sgate_i\}_{i=1}^n\,\rangle.$$

We leverage this symmetry by noting that the distinct distillation protocols correspond to the double cosets $\mathcal{D}_n\backslash \Sp(n,\F_2)/\mathcal{K}_n$, similar to our argument before for right cosets for general input states. In this section, we describe how one can rewrite an arbitrary symplectic matrix $M$ to another symplectic matrix $M'$ of a specific form, which is in the same double coset as $M$. Such a representative $M'$ of the double coset has a smaller number of free parameters, reducing the search space significantly.

Recall that an overview of the basic matrix operations corresponding to elements of $\Sp(2n,\mathbb{F}_2^{2n})$ is given in Appendix \ref{sec:matrixoperations}.

\begin{lemma}\label{lemma:nfoldwernerreduction}
Let $M\in \Sp(2n,\F_2)$. There exists $M'$ in the double coset $\mathcal{D}_nM\mathcal{K}_n$ that is of the form 
$$
M'=\begin{bmatrix}A&B\\0&A^\transp\end{bmatrix},\  A=\left[\begin{array}{c|c}1&0\\\hline &  \\a&\:I_{n-1}\:\\& \end{array}\right],
$$
$$
B=\left[\begin{array}{c|c}0&b^\transp\\\hline & \\b&E\!+\!ba^\transp\!\\ &\end{array}\right],
$$
where $a,b\in \F_2^{n-1}$ and $E\in \F_2^{(n-1)\times (n-1)}$ is symmetric with zeroes on the diagonal.  
\end{lemma}

\begin{proof}
Let $M'\in \mathcal{D}_nM\mathcal{K}_n$ be such that  
\begin{equation}\label{step1}
M'_{ij}=\delta_{ij}\quad\text{for $(i,j)\in [n]\times \{2,\ldots, k\}$}
\end{equation}
with $1\leq k\leq n$ as large as possible. Note that for $k=1$ the condition is trivially fulfilled. 

\begin{claim}
$k=n$.
\end{claim}
\begin{claimproof}
Suppose that $k<n$. Then $M'_{k+1,k+1}=0$, otherwise we can use row operations on $M'$ (left-multiplication by matrices $\CNOT_{k+1,i}\in \mathcal{D}_n$) to obtain $M'_{i,k+1}=\delta_{i,k+1}$ for all $i\in[n]$ while keeping \eqref{step1}, contradicting the maximality of $k$.

Note that the above condition $M'_{k+1,k+1}=0$ needs to hold after applying operations to $M'$ that preserve the form in equation \eqref{step1}.
Thus, by permuting rows in $\{k+1,\ldots, n\}$ (left-multiplication by matrices $\SWAP_{ij}\in \mathcal{D}_n$) or permuting columns in $\{1,k+1,\ldots, n\}$ (right-multiplication by matrices $\SWAP_{ij}\in \mathcal{K}_n$) we deduce that $M'_{ij}=0$ for $(i,j)\in \{k+1,\ldots, n\}\times \{1,k+1,\ldots, n\}$. Since we can swap column $i$ and $n+i$ by multiplying from the right with $H_i \in \mathcal{K}_n$, we also have $M'_{ij}=0$ for $(i,j)\in \{k+1,\ldots, n\}\times \{n+1,n+k+1,\ldots, 2n\}$. To summarise, we have
\begin{eqnarray*}
M'_{ij}&=&0\quad\text{for $i\in \{k+1,\ldots, n\}$}\\
&&\text{and $j\in [2n]\setminus\{n+2,\ldots,n+k\}$}\\
M'_{ij}&=&\delta_{ij}\quad \text{for $i\in \{2,\ldots, k\}$ and $j\in \{2,\ldots, k\}$}.
\end{eqnarray*}
Since rows $k+1,\ldots, n$ must have zero symplectic inner product with rows $2,\ldots, k$, it follows that rows $k+1,\ldots, n$ must in fact be equal to zero. Since $M'$ has full rank, this implies that $k=n$.~\end{claimproof}

Consider the first row of $M'$. We have $M'_{1,j}=0$ for $j=2,\ldots, n$. If $M'_{1,1}=M'_{1,n+1}=0$, then the fact that this row has zero symplectic inner product with rows $2,\ldots, n$ implies that the first row is equal to zero, which is not possible as $M'$ has full rank. So by multiplying from the right by $I$, $\Hgate_1$, or $\Sgate_1\Hgate_1$ which are in $\mathcal{K}_n$, we may assume that $M'_{1,1}=1$ and $M'_{1,n+1}=0$. 

Writing $M'=\big[\begin{smallmatrix}A&B\\C&D\end{smallmatrix}\big]$, we see that $A$ and $B$ have the  following form:
$$
A=\left[\begin{array}{c|c}1&0\\\hline &  \\a&\:I_{n-1}\:\\& \end{array}\right], \quad B=\left[\begin{array}{c|c}0&d^\transp\\\hline & \\b&\ \ E'\ \ \\ &\end{array}\right].
$$
Since row $1$ has zero symplectic inner product with rows $2,\ldots, n$, it follows that $d=b$. Note that for $1\leq i<j\leq n-1$ the symplectic inner product of rows $i+1$ and $j+1$ is equal to $a_ib_j+E'_{ji}+a_jb_i+E'_{ij}$. Since this inner product is zero, the matrix $E:=E'+ba^\transp$ is symmetric. By multiplying from the right by $\Hgate_i\Sgate_i\Hgate_i\in \mathcal{K}_n$ ($i=2,\ldots, n$) if necessary, we may set the diagonal elements of $E'$ such that the diagonal elements of $E$ are zero.

Recall from Lemma \ref{thm:distsubgroupgen} that $\Sgate_{ij}:=(\Sgate_j\CNOT_{ij})^2\in \mathcal{D}_n$ for $1\leq j<i$. Recall furthermore that left-multiplication of $M'$ by $\Sgate_{ij}$ amounts to adding row $i$ to row $n+j$ and adding rows $i$ and $j$ to row $n+i$. By left-multiplication by elements $\Sgate_{ij}\in \mathcal{D}_n$ and $\Sgate_i\in \mathcal{D}_n$ we may (without changing the first $n$ rows of $M'$) assume that $C$ is a strictly upper triangular matrix. Since the first $n$ columns of $M'$ must have pairwise zero symplectic inner product, this implies that in fact $C=0$. Since $A^\transp D+C^\transp B=I_n$, it follows that $D=(A^\transp)^{-1}=A^\transp$, where we have used that $A$ is self-inverse.

\end{proof}

Note that for any permutation $\pi\in S_{n-1}$, we can replace $a,b,E$ by $\pi(a),\pi(b),\pi(E)$ (permuting both rows and columns) by multiplying $M'$ simultaneously from the left and the right by elements $\SWAP_{ij}$, since $\SWAP_{ij}$ is an element of $\mathcal{D}_n$ and $\mathcal{K}_n$ for $2\leq i<j\neq n$. Also, we can replace $(a,b)$ by $(b,a)$ or $(a,a+b)$ by multiplication from the left and right by elements from $\{\Sgate_1,\Hgate_1\}$. Hence, to cover all cases, it suffices to enumerate over the triples $(a,b,E)$ where $a\leq b\leq a+b$ and $E$ runs over the adjacency matrices of graphs on $n-1$ nodes (up to isomorphism).

\section{Optimisation results}\label{sec:results}
In the previous sections we have outlined our methods for finding all possible bilocal Clifford protocols, which were described in Section \ref{sec:bilocalClifford}. In the following we report our findings, first for up to $n=5$ general Bell-diagonal input states, second for up to $n=8$ identical Werner states.

\subsection{Achieved distillation statistics for general input states}
In Fig.~\ref{fig:fidelitypsucscatter} we show the achievable $\left(p_\textrm{suc},~F_\textrm{out}\right)$ pairs for $n = 2, 3, 4, 5$ copies of a state with coefficients $p_I = 0.7$, $p_X = 0.15$, $p_Y = 0.10$, $p_Z = 0.05$. We also plot the envelope, indicating the best performing schemes. Moreover, our results for $n = 5$ clearly show that picking an arbitrary coset does not give a good protocol in general.

Furthermore, we consider the $n=2$ scenario where the two input states are equal, i.e.~both have equal values for the $p_I, p_X, p_Y, p_Z$ parameters. By comparing all analytic expressions of the output fidelity as a function of $p_I, p_X, p_Y, p_Z$, we find that the DEJMPS protocol achieves the highest fidelity out of all bilocal Clifford protocols (see~\cite{code} for the details). 

While we do not explore this direction, the results can also be applied to less symmetric cases, i.e.~when the $n$ pairs are not the same. This situation is, for example, relevant when states arrive at different times, and thus experience different amounts of decoherence. 

\begin{figure}
\centerfloat
	\includegraphics[clip,  width=0.5\textwidth, trim = 2.5mm 3.0mm 2.8mm 0mm]{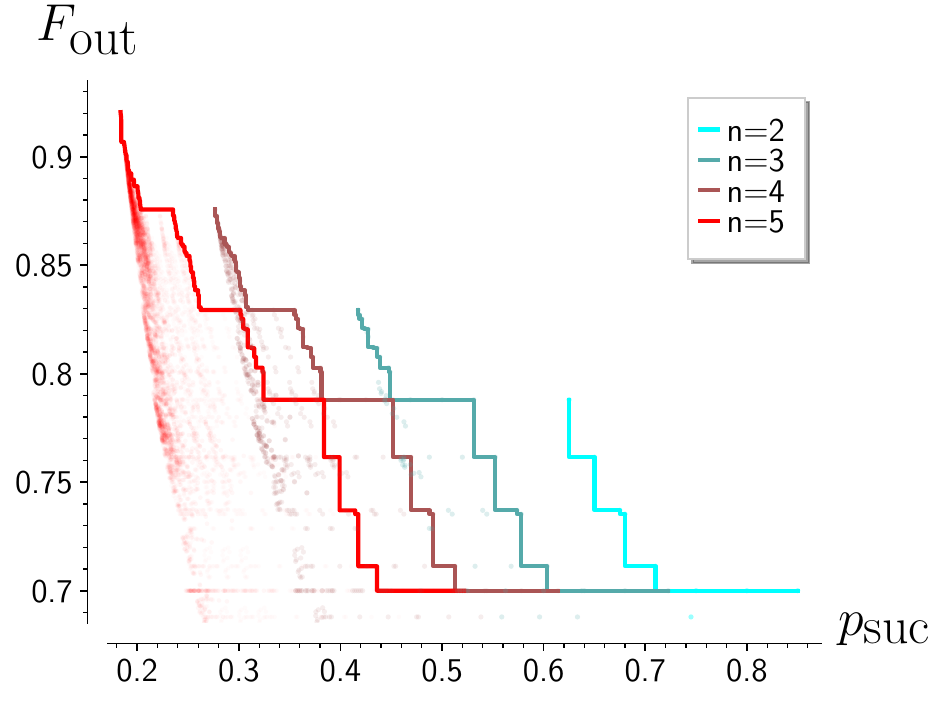}
	\vspace*{-3mm}
	\caption{Achievable $\left(p_\textrm{suc},~F_\textrm{out}\right)$ pairs for $n = 2$ to $5$ copies of a state with coefficients $p_I = 0.7$, $p_X = 0.15$, $p_Y = 0.10$, $p_Z = 0.05$. The highest achievable pairs are indicated by a solid line for each number of copies. Not included in the plot are those distillation protocols with fidelity smaller than $F=0.68$.}
	\label{fig:fidelitypsucscatter}
\end{figure}

\subsection{Achieved distillation statistics for $n$-fold Werner states}
Here we show our results for the case of an $n$-fold tensor product of a Werner state. 
First, we list the number of cases to check (i.e.~the number of triples $(a, b, E)$, see Section~\ref{sec:wernerstatesection}) and the number of distinct distillation protocols for this scenario in Table~\ref{table:wernercases}.

\begin{table*}[t]
\centerfloat
\begin{tabular}{l || c | c | c | c | c | c | c}
 & 2 & 3 & 4 & 5 & 6 & 7 & 8 \\
 &&&&&&\\[-0.8em]
\hline
&&&&&&\\[-0.8em]
Cases & 2 & 10 & 60 & 561 & 6358 & 111540 & 2917980 \\
&&&&&&\\[-0.8em]
\hline
&&&&&&\\[-0.8em]
Distinct protocols & 2 & 5 & 13 & 34 & 108 & 379 & 1736 
\end{tabular}
\caption{Number of cases to check and number of distinct distillation protocols for $n$ identical Werner states.}
\label{table:wernercases}
\end{table*}

The number of cases and distinct protocols still rapidly grow with $n$, but our reduction allowed to consider all possible distillation protocols for $n=8$ in about a day of computer run-time. This should be compared with a naive optimisation over all elements of the symplectic group --- for $n=8$ the ratio between the order of the symplectic group and the number of cases to check is approximately $2\cdot 10^{34}$. For $n=9$ and higher however, our current method becomes infeasible, requiring too many cases to consider. One could imagine improving on Lemma V.1 as to reduce the gap between the number of cases to check and the number of distinct distillation protocols. This could potentially only allow for an enumeration up to $n\approx 11$. This can be made more quantitative by considering that a lower bound on the number of double cosets is given by $\frac{\left|\Sp(2n,\mathbb{F}_2)\right|}{\left|\mathcal{D}_n\right|\left|\mathcal{K}_n\right|} = \mathcal{O}\left(\frac{2^{\frac{n^2}{2}}}{6^n n!}\right)$, since the size of a double coset can be at most $\left|\mathcal{D}_n\right|\cdot \left|\mathcal{K}_n\right|$. We note here that we are approximating the number of distinct distillation statistics by the number of double cosets. With such a lower bound, it is clear that a full enumeration becomes infeasible for $n\gtrapprox 11$, even in the best-case scenario.

We note here that the large number of distinct protocols only means that a full enumeration of the distillation protocols becomes infeasible for larger $n$. In particular, it does not necessarily preclude an optimisation of the distillation protocols for a given metric, such as the output fidelity. In this paper, however, we only consider an optimisation by first fully enumerating all distinct protocols.

In order to gauge the advantage of the optimal protocols that we find for Werner states, we compare them with the class of protocols we call \emph{concatenated DEJMPS protocols}. These are bilocal Clifford protocols that are built from multiple iterations of the DEJMPS protocol~\cite{Deutsch1996}, see Appendix~\ref{sec:dejmpsprotocols} for more information. The concatenated DEJMPS protocols form a natural generalisation of the (nested) entanglement pumping protocols~\cite{dur2007entanglement}.

We first investigate the increase in fidelity $F_\textrm{out}-F_\textrm{in}$ conditioned on the success of the distillation protocol. 
We plot the increase in fidelity as a function of the input fidelity $F_\textrm{in}$ for $n= 2,3,\ldots, 8$ in Fig.~\ref{fig:fidelity1}. The dotted lines correspond to the concatenated DEJMPS protocols, the solid lines correspond to the protocols that achieve the highest output fidelity found with our optimisation.
For completeness, we show the success probabilities and fidelities for the optimised protocols for $n = 2, 3,\ldots, 8$ in Tables~\ref{table:distillationstatistics}, \ref{table:distillationstatistics2}, \ref{table:distillationstatisticsinfidelity} and \ref{table:distillationstatistics2infidelity} in Appendix~\ref{sec:diststatistics}.

\begin{figure}
\centerfloat
	\includegraphics[clip,  width=0.5\textwidth, trim = 3.0mm 3.0mm 3.1mm 0mm]{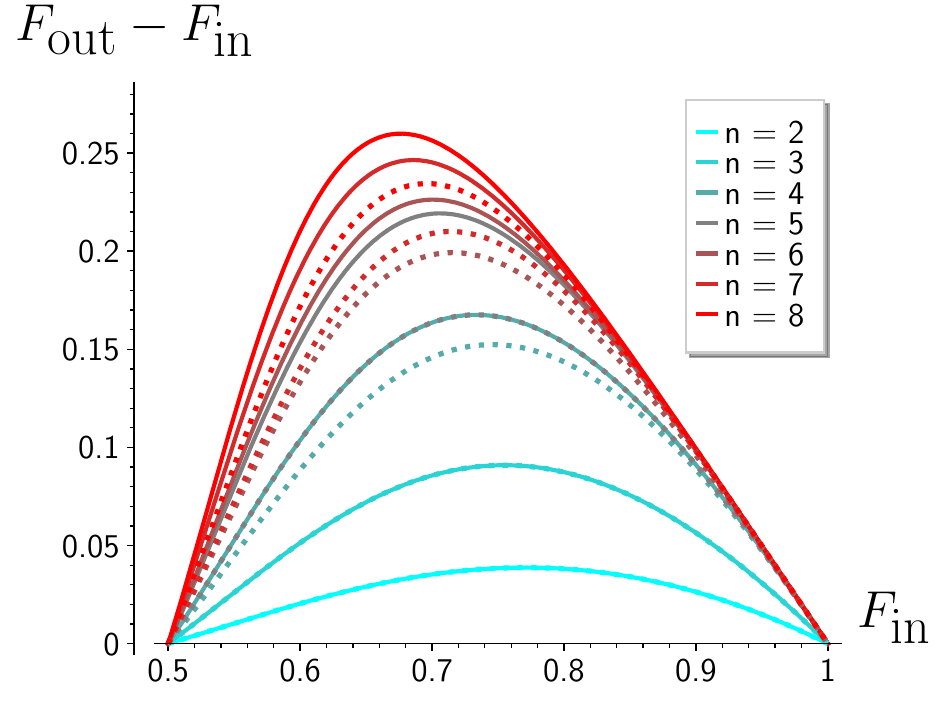}
	\vspace*{-3mm}
	\caption{Comparison between the increase in fidelity $F_\textrm{out}-F_\textrm{in}$ with our optimisation (solid) and concatenated DEJMPS protocol (dotted), for $n=2$ to $8$ identical Werner states with fidelity $F_\textrm{in}$. Note how the $n=5$ concatenated DEJMPS protocol overlaps with an optimised $n=4$ protocol.}
	\label{fig:fidelity1}
\end{figure}

Let us now discuss Fig.~\ref{fig:fidelity1}. First we observe that for $n=2, 3$, the optimal protocols correspond to the original DEJMPS~\cite{Deutsch1996} and double-selection~\cite{fujii2009entanglement} protocols. However, for $n>3$, we find distillation protocols that outperform the concatenated DEJMPS protocols. 

We find that the optimal protocol for $n=4$ achieves the same fidelity as the concatenated DEJMPS protocol for $n=5$, and can be executed with a circuit of the same depth as the concatenated DEJMPS protocol. This protocol achieves the same distillation statistics as the protocol found with different means in the recent work from~\cite{zhao2021loccnet}. 

For $n=5$ there is a large gap between the optimised protocols and the concatenated DEJMPS protocol. We make this now more quantitative by expanding the $F_\textrm{out}$ for high input fidelity $F_\textrm{in}\approx 1$. For $n=5$, the concatenated DEJMPS protocol has quadratic scaling in the infidelity,

\begin{align}
    1-\frac{2}{3}\left(1-F\right)^2 + \mathcal{O}\left(\left(1-F\right)^3\right)\ ,
\end{align}

while the optimised protocol has a cubic scaling in the infidelity
\begin{align}
    1-\frac{10}{9}\left(1-F\right)^3 + \mathcal{O}\left(\left(1-F\right)^4\right)\ .
\end{align}

This is particularly surprising, since previous protocols with five or less pairs~\cite{krastanov2019optimized} have a scaling that is at most quadratic in the infidelity. We list the scaling of the found protocols in Table.~\ref{table:distillationstatisticsinfidelityseries}.

Next, we investigate the behavior of the protocols for high fidelities $F_\textrm{in} \approx 1$. In Fig.~\ref{fig:fidelity2} we plot the infidelity $1-F_\textrm{out}$ as a function of the input fidelity $F_\textrm{in}$. We observe that it is possible to reach fidelities of around $0.999$ using six copies of Werner states with fidelity $F_\textrm{in} = 0.9$. We do not plot the results for $n = 2, 3$ since we find no improvements with respect to previous protocols.

\begin{figure}
\centerfloat
	\includegraphics[clip,  width=0.5\textwidth, trim = 3.0mm 3.0mm 3.1mm 0mm]{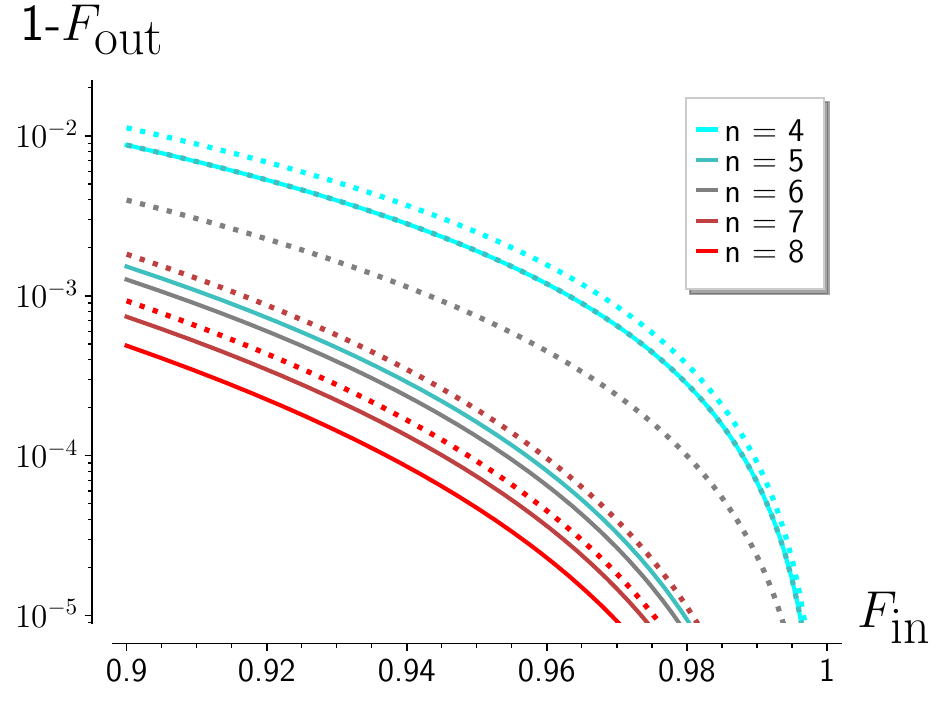}
	\vspace*{-3mm}
	\caption{Comparison between the achieved infidelities $1-F_\textrm{out}$ with our optimisation (solid) and concatenated DEJMPS protocol (dotted), for $n=4$ to $8$ identical Werner states with fidelity $F_\textrm{in}$.}
	\label{fig:fidelity2}
\end{figure}

We have seen that the optimised distillation protocols are capable of achieving a higher fidelity than the concatenated DEJMPS protocols. However, the optimised distillation protocols also have a lower success probability. This motivates us to investigate three metrics, each of which combines the success probability and the quality of the resultant state. As the first metric, we use the distillable entanglement rate which we approximate by combining the distillation protocol together with a \emph{hashing} protocol~\cite{bennett1996purification}. That is, given $n$ entangled pairs, we first perform an $n$-to-1 distillation protocol and then use the output as input for the hashing protocol.
The rate $r$ at which this procedure produces maximally entangled state is given by

\begin{align}
    r = \frac{\left(1-H(p)\right)\cdot p_\textrm{suc}}{n}\ ,
\end{align}
where $H(p)$ is the entropy in bits of the probability distribution $p = \left( p_\mathrm{I}, p_X, p_Y, p_Z\right)$ corresponding to the output state. This metric has been used previously and is sometimes called the hashing yield~\cite{krastanov2019optimized}.

We compare the achieved rate of all found distillation protocols and the concatenated DEJMPS protocols in Fig.~\ref{fig:rate1}. We show both achieved rates and the ratio between them. We find that for $n>3$ and fidelities $\lesssim 0.76$ the optimal bilocal Clifford protocols achieve up to rates three times greater than concatenated DEJMPS protocols. Conversely, for high fidelities it suffices to use concatenated DEJMPS protocols if one is interested in maximising the asymptotic rate.

For the second metric, we consider the more practical application of achieving a certain threshold fidelity $F_\textrm{tar}$ using (at most) one round of distillation. While there exist a number of applications that require a minimum fidelity, we consider here device-independent quantum key distribution. If one assumes only depolarising noise, a bound on the minimum fidelity is given~\cite{tan2020improved} by $F_\textrm{tar} = 0.930025$ to perform device-independent quantum key distribution. In what follows, we will assume that it is necessary to achieve $F_\textrm{tar}$, and that the other coefficients of the output state are equal.

In Fig.~\ref{fig:rate_distill_once} we show the average rate at which entanglement can be distilled to $F_\textrm{tar}$ using a single round of distillation. We find a similar behaviour as in Fig.~\ref{fig:rate1}, where for higher fidelities ($\approx 0.78$) it suffices to perform concatenated DEJMPS protocols, while for lower fidelities the optimised protocols achieve a larger rate.

Finally, we consider the metric of the success probability times the \emph{relative entropy of entanglement}~\cite{vedral1997quantifying, vedral1998entanglement, vedral2002role}. The relative entropy of entanglement (REE) is an upper bound on the distillable entanglement~\cite{vedral2002role}, and can be computed exactly for Bell-diagonal states~\cite{vedral1997quantifying}. The success probability times the REE is a quantity that has been used to capture how well a protocol concentrates the present entanglement, and has even been used to show optimality of some distillation protocols, see~\cite{Rozpdek2018}. We note that we have found that the success probability times the REE decreases as the number of input states increases. We thus consider maximising the above quantity for fixed values of $n$. To this end, we plot the difference between the success probability and the REE for all optimised protocols and concatenated DEJMPS protocols in Fig.~\ref{fig:rel_ent} for $n = 2,\ldots, 8$. We find for all $n>3$ that the full optimisation over bilocal Clifford protocols allows for a larger value of the success probability times the REE, in particular for higher fidelities. This should be contrasted with the results from Figs.~\ref{fig:rate1} and~\ref{fig:rate_distill_once}, where the full optimisation only showed an improvement over concatenated DEJMPS protocols for lower fidelities.

\begin{figure}
\centerfloat
\begin{subfigure}[b]{0.52\textwidth}
	\includegraphics[clip,  width=0.996\textwidth, trim = 3.0mm 3.0mm 3.1mm 0mm]{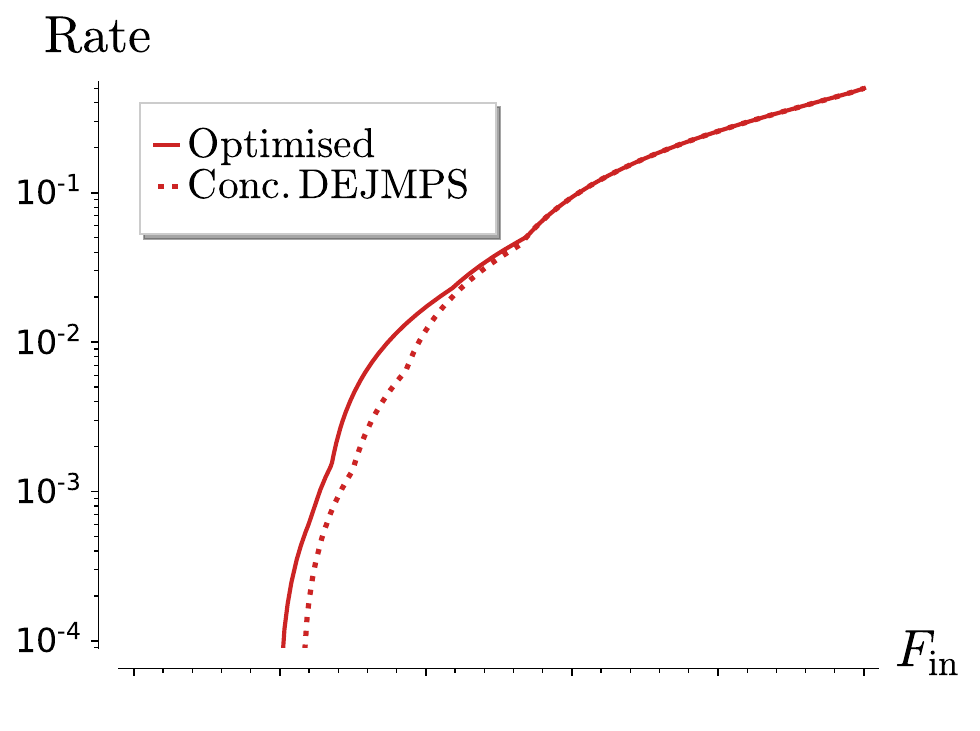}
	\vspace*{-8mm}
\end{subfigure}

\begin{subfigure}[b]{0.5\textwidth}
	\includegraphics[clip,  width=0.978\textwidth, trim = 3.0mm 3.0mm 3.1mm 0mm]{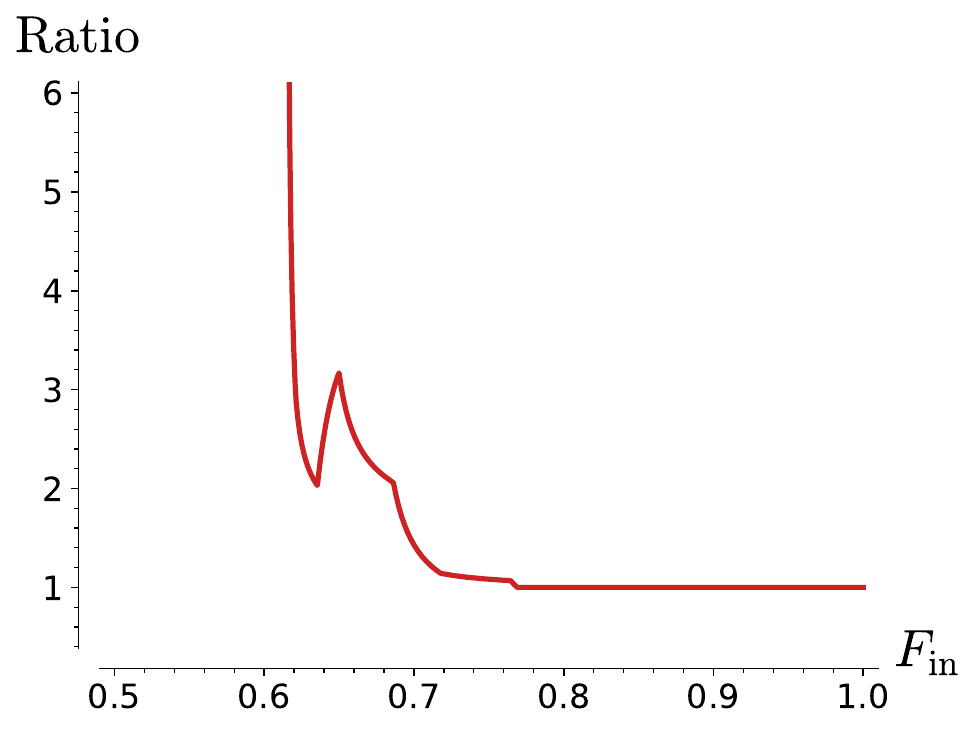}
	\vspace*{-3mm}
\end{subfigure}
\caption{Comparison between the achieved rates after distilling and then hashing with our optimisation (solid) and the concatenated DEJMPS (dotted) protocol. For both cases, we take the envelope of all protocols on $n=2$ to $8$ identical Werner states with fidelity $F_\textrm{in}$. Top) achieved rates with our optimisation and concatenated DEJMPS protocols. Bottom) Ratio between the rate with our optimisation and concatenated DEJMPS protocols.}
\label{fig:rate1}
\end{figure}

Let us conclude with an investigation of circuits that achieve the highest fidelity for $n=4$ to $8$.
Interestingly, these protocols can be implemented with low-depth circuits. We performed a search over circuits of the form described in Appendix \ref{sec:distcircuits}, to find circuits that achieve the highest fidelity. We report these circuits in Appendix~\ref{sec:distcircuits}. For $n=4$ to $8$, we find a total number of two-qubit gates of $4$, $7$, $8$, $11$ and $13$. Furthermore, the corresponding circuit depths are $3$, $5$, $6$, $6$ and $7$, respectively. For comparison, the circuit from~\cite{zhao2021loccnet} for $n=4$ pairs has $4$ two-qubit gates and depth $5$. This protocol can be converted to our optimal $n=4$ protocol by left-multiplication with elements in $\mathcal{D}_n$ and right-multiplication  with elements in $\mathcal{K}_n$. Therefore, both protocols achieve the same distillation statistics. The protocol from~\cite{matsuo2019quantum} for $n=5$ pairs, which achieves the same fidelity and success probability as the concatenated DEJMPS protocol, has $4$ two-qubit gates and depth $4$.
We note here that the fidelity and success probability do not necessarily need to correspond to a unique specific distillation protocol. As an example, for $n=8$ there are four distinct protocols that achieve the highest fidelity. These protocols have the same fidelity and success probability, but differ in their $F_i$ components.

\begin{figure}
\centerfloat
	\includegraphics[clip,  width=0.5\textwidth, trim = 3.0mm 3.0mm 3.1mm 0mm]{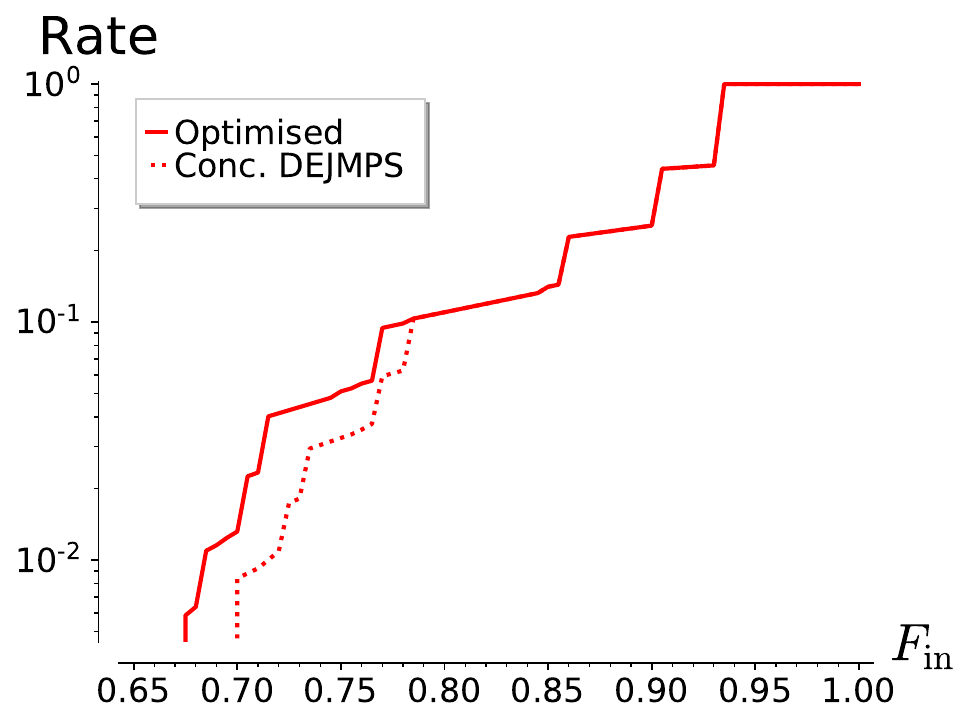}
\caption{Average rate at which entanglement can be distilled to a minimum fidelity of $F_{\textrm{tar}} = 0.930025$, using both our optimisation (solid) and concatenated DEJMPS (dotted) protocols.}
\label{fig:rate_distill_once}
\end{figure}

\begin{figure}
\centerfloat
	\includegraphics[clip,  width=0.5\textwidth, trim = 3.0mm 3.0mm 3.1mm 0mm]{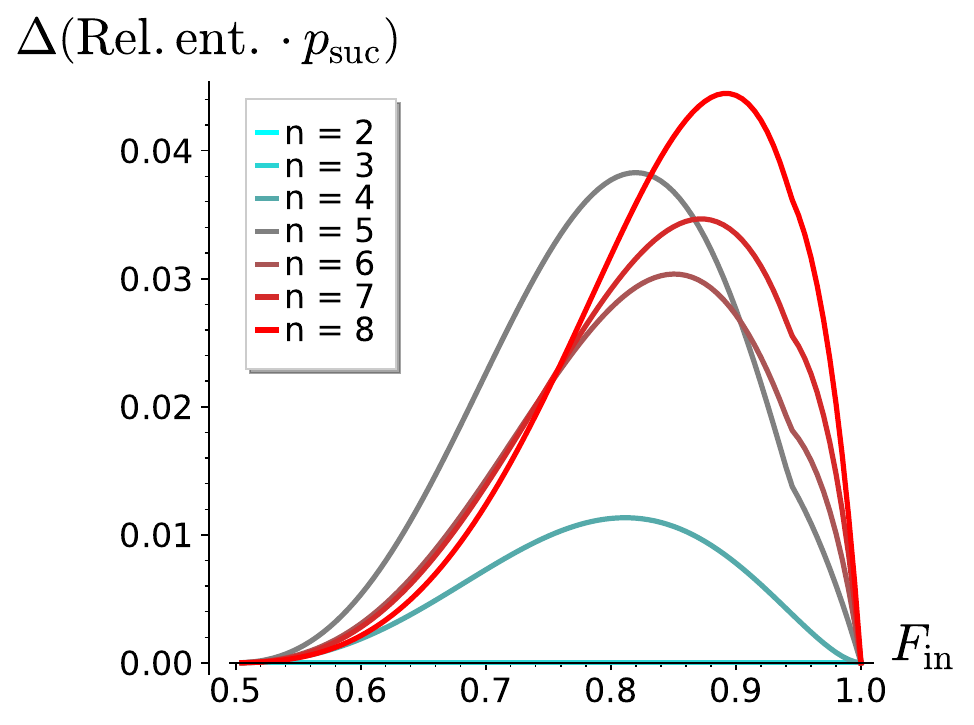}
\caption{Difference between the product of the success probability and relative entropy of entanglement of the resultant state, optimised over all bilocal Clifford protocols and all concatenated DEJMPS protocols.}
\label{fig:rel_ent}
\end{figure}

\section{Conclusions and discussions}\label{sec:conc}
Our goal in this paper was to find good distillation protocols requiring modest resources. For this, we introduced the class of bilocal Clifford protocols which generalises many existing protocols. The protocols in this class require only a single round of communication between the end parties and the implementation of Clifford gates. Within this class, we leveraged group theoretic tools to find all distillation protocols for up to $n=5$ pairs for general Bell-diagonal states and up to $n=8$ pairs for the $n$-fold tensor product of a Werner state. 

Some of the protocols that we found strongly improve upon the fidelities and rates of previous protocols.
Moreover, we give explicit circuits for the optimal protocols for the $n$-fold Werner state case, with $n=2$ to $n=8$. These circuits have comparable depth and number of two-qubit gates as previous protocols, indicating the experimental feasibility of the new protocols. If the improved performance holds with noisy operations, then it will translate in improved forecasts for the performance of near-term quantum networks~\cite{matsuo2019quantum,victora2020purification} or distributed quantum computation~\cite{nickerson2013topological}. 
Finally, since we have enumerated all bilocal Clifford protocols up to $n=5$, it is possible to pick and choose the protocol that maximises any figure of merit for any particular set of input states. Our software can be found at~\cite{code}.

In this work we considered distilling one entangled pair out of $n$ pairs. The results here could be extended to $n$ to $m$ distillation protocols by generalising Lemma~\ref{lemma:nfoldwernerreduction} and the characterisation of the distillation subgroup to the case of $n$ to $m$ distillation. Such distillation protocols would allow for a more refined trade-off between the fidelity and the success probability/rate. Another interesting avenue would be to generalise the tools to higher dimensional entangled states. 

\section*{Acknowledgements}
This work was supported by the QIA project (funded by European Union's Horizon 2020, Grant Agreement No. 820445) and by the Netherlands Organization for Scientific Research (NWO/OCW), as part of the Quantum Software Consortium program (Project No. 024.003.037/3368). The authors thank T.~Coopmans for initial discussions, and F.~Rozp\k{e}dek and S.~B\"{a}uml for valuable feedback on the manuscript. The authors thank Pol Julià Farré for pointing out a typo in Table 6.

\bibliographystyle{quantum}
\bibliography{references}
\newpage
\onecolumngrid
\appendix

\section{Background on binary picture}
\label{sec:backgroundbinary}
For completeness, we give here more background and derivations on the binary picture used in this work. Firstly, we give a derivation of equation \eqref{eq:paulimatrices}. Suppose that we have two elements $\tau_{v}, \tau_{w} \in \mathcal{\thickbar P}_n$, with $v, w \in \mathbb{F}_2^{2n}$. Then
\begin{equation}
    \begin{split}
    \tau_{v}\tau_{w} & = (\tau_{v_1v_{n+1}}\otimes\dots\otimes \tau_{v_nv_{2n}})(\tau_{w_1w_{n+1}}\otimes\dots\otimes \tau_{w_nw_{2n}})\\
    & = \bigotimes_{k=1}^n \tau_{v_kv_{n+k}}\tau_{w_kw_{n+k}}.
    \end{split}
\end{equation}

For all $k \in \{1,\dots,n\}$, we have
\begin{equation}
\begin{split}
    \tau_{v_kv_{n+k}}\tau_{w_kw_{n+k}}
    & = X^{v_k}Z^{v_{n+k}}X^{w_k}Z^{w_{n+k}} \\
    & = X^{v_k}(-1)^{v_{n+k}w_k}X^{w_k}Z^{v_{n+k}}Z^{w_{n+k}} \\
    & = (-1)^{v_{n+k}w_k}X^{v_k+w_k}Z^{v_{n+k}+w_{n+k}}\\
    & = (-1)^{v_{n+k}w_k}\tau_{v_k+w_k,v_{n+k}+w_{n+k}}.
    \end{split}
\end{equation}

As a result,
\begin{equation}
    \begin{split}
    \tau_{v}\tau_{w} &= \bigotimes_{k=1}^n (-1)^{v_{n+k}w_k}\tau_{v_k+w_k,v_{n+k}+w_{n+k}} \\
    & = (-1)^{\sum_{k=1}^n v_{n+k}w_k}\ \tau_{v+w}
    \end{split}
\end{equation}

We can rewrite $\sum_{k=1}^n v_{n+k}w_k$ in terms of the vectors $v$ and $w$: 
\begin{equation}
    \sum_{k=1}^n v_{n+k}w_k = w^\transp \Xi v, \qquad \ \Xi = \begin{bmatrix} 0 & I_n \\ 0 & 0 \end{bmatrix}.
\end{equation}

Hence, the product of $\tau_v$ and $\tau_w$ is given by
\begin{equation}
    \tau_{v}\tau_{w} = (-1)^{w^\transp \Xi v}\tau_{v + w}.
    \label{eq:binaryrepprod}
\end{equation}

Combining equation \eqref{eq:binaryrepprod} for $\tau_v\tau_w$ and $\tau_v\tau_w$, we finally obtain
\begin{equation}
    \tau_v\tau_w = (-1)^{w^\transp\Xi v + v^\transp\Xi w}\tau_w\tau_v = (-1)^{v^\transp\Xi^\transp w + v^\transp\Xi w}\tau_w\tau_v = (-1)^{v^\transp \Omega w}\tau_w\tau_v, \quad \Omega = \Xi + \Xi^\transp = \begin{bmatrix} 0 & I_n \\ I_n & 0 \end{bmatrix}.
\end{equation}

Let $C \in \mathcal{C}_n$ be a Clifford operation and $f: \mathcal{\thickbar P}_n \rightarrow \mathcal{\thickbar P}_n$, $f(P) = CPC^\dagger$ be the corresponding automorphism. Let $\pi: \F^{2n}_2 \rightarrow \F^{2n}_2$ be the representation of $f$ in the binary picture. Let $v, w \in \F_2^{2n}$.
Then we know that $C\tau_{v+w}C^\dagger = (-1)^{w^\transp\Xi v} C\tau_{v}\tau_{w}C^\dagger = (-1)^{w^\transp\Xi v}C\tau_{v}C^\dagger C\tau_{w}C^\dagger$. In the binary representation, the prefactor $(-1)^{w^\transp\Xi v}$ does not make a difference. Thus, $\pi(v+w) = \pi(v) + \pi(w)$, so $\pi$ is a linear map, and there exists a binary $2n \times 2n$ matrix $M$ such that $\pi(v) = Mv$ for all $v \in \F_2^{2n}$.

\section{Basic row/column operations corresponding to symplectic matrices}
\label{sec:matrixoperations}
In this section we give an overview of the basic row and column operations used in the proofs of Lemma \ref{thm:distsubgroupgen}, Theorem \ref{thm:orderDn} and Lemma \ref{lemma:nfoldwernerreduction} and the corresponding elements of $\Sp(2n, \mathbb{F}_2)$.
Let $M \in \text{Mat}_{2n \times k}(\mathbb{F}_2)$ be a binary $(2n \times k)$-matrix with $k \geq 1$. Multiplying $M$ from the left by an element of $\Sp(2n, \mathbb{F}_2)$ results in basic row operations on the rows of $M$. The basic row operations corresponding to left multiplication by elements of $\Sp(2n, \mathbb{F}_2)$ are summarized in Table \ref{tab:rowoperations} \cite{Maslov2018}.

\begin{table}[h]
    \centering
    {\begin{tabular}{l|l}
    Element of $\Sp(2n, \mathbb{F}_2)$ \ & Row operation \\ &\\[-0.8em]
	\hline \hline
	&\\[-0.8em]
    $H_i$ & Swapping rows $i$ and $n+i$\\
    &\\[-0.8em]
	\hline
	&\\[-0.8em]
    $S_i$ & Adding row $i$ to row $n+i$\\
    &\\[-0.8em]
	\hline
	&\\[-0.8em]
    $\CNOT_{ij}$ & Adding row $i$ to row $j$ and adding row $n+j$ to row $n+i$\\
    &\\[-0.8em]
	\hline
	&\\[-0.8em]
	$S_{ij} = (S_j\CNOT_{ij})^2$ & Adding row $i$ to row $n+j$ and adding rows $i$ and $j$ to row $n+i$\\
	&\\[-0.8em]
	\hline
	&\\[-0.8em]
	$\CZ_{ij}$ & Adding row $i$ to row $n+j$ and adding row $j$ to row $n+i$\\
	&\\[-0.8em]
	\hline
	&\\[-0.8em]
	$\SWAP_{ij}$ & Swapping rows $i$ and $j$ and swapping rows $n+i$ and $n+j$
\end{tabular}}
    \caption{Basic row operations corresponding to left multiplication of $M \in \text{Mat}_{2n \times k}(\mathbb{F}_2)$ by the elements of $\Sp(2n, \mathbb{F}_2)$ in the first column.}
    \label{tab:rowoperations}
\end{table}

Similarly, multiplying $M \in \text{Mat}_{k \times 2n}(\mathbb{F}_2)$ from the right by an element of $\Sp(2n, \mathbb{F}_2)$ results in basic column operations on the columns of $M$. These column operations are summarized in Table \ref{tab:columnoperations}.

\begin{table}[h]
    \centering
    {\begin{tabular}{l|l}
    Element of $\Sp(2n, \mathbb{F}_2)$ \ & Column operation \\ &\\[-0.8em]
	\hline \hline
	&\\[-0.8em]
    $H_i$ & Swapping columns $i$ and $n+i$\\
    &\\[-0.8em]
	\hline
	&\\[-0.8em]
    $S_i$ & Adding column $n+i$ to column $i$\\
    &\\[-0.8em]
	\hline
	&\\[-0.8em]
    $\CNOT_{ij}$ & Adding column $j$ to column $i$ and adding column $n+i$ to column $n+j$\\
    &\\[-0.8em]
	\hline
	&\\[-0.8em]
	$S_{ij} = (S_j\CNOT_{ij})^2$ & Adding columns $n+i$ and $n+j$ to column $i$ and adding column $n+i$ to \\ & column $j$\\
	&\\[-0.8em]
	\hline
	&\\[-0.8em]
	$\CZ_{ij}$ & Adding column $n+i$ to column $j$ and adding column $n+j$ to column $i$\\
	&\\[-0.8em]
	\hline
	&\\[-0.8em]
	$\SWAP_{ij}$ & Swapping columns $i$ and $j$ and swapping columns $n+i$ and $n+j$
\end{tabular}}
    \caption{Basic column operations corresponding to right multiplication of $M \in \text{Mat}_{k \times 2n}(\mathbb{F}_2)$ by the elements of $\Sp(2n, \mathbb{F}_2)$ in the first column.}
    \label{tab:columnoperations}
\end{table}

Note that we are particularly interested in the cases $k = 2n$ (Lemma \ref{thm:distsubgroupgen}, Lemma \ref{lemma:nfoldwernerreduction}) and $k = 1$ (Theorem \ref{thm:orderDn}), although Table \ref{tab:rowoperations} and Table \ref{tab:columnoperations} hold true for any $k \geq 1$. 

\section{Concatenated DEJMPS protocols}\label{sec:dejmpsprotocols}
Here we describe the distillation protocols which we compare our results with. These are all based on the so-called DEJMPS protocol~\cite{Deutsch1996}. The DEJMPS protocol takes two pairs of Bell-diagonal states, and outputs one state. It performs bilocal single-qubit rotations on both pairs, then a bilocal CNOT, and finally a measurement on one of the pairs where a success is achieved only when correlated outcomes are observed. It is clear that the DEJMPS protocol is an example of a bilocal Clifford protocol. The DEJMPS protocol can be generalised to a number of pairs $n>2$ by applying the DEJMPS protocol multiple times.

Since the DEJMPS protocol corresponds to 2-1 distillation, the possible ways of combining the different pairs correspond to the number of binary trees on $n$ unlabeled nodes for an $n$-fold tensor product of input states. Furthermore, for each of the performed DEJMPS protocols (corresponding to each parent of the binary tree), we consider all possible single-qubit rotations. The \emph{concatenated DEJMPS protocols} are then all protocols that arise in this fashion. Note that this class includes well known variants of DEJMPS such as (nested) entanglement pumping protocols~\cite{dur2007entanglement, dur1999quantum, briegel1998quantum} or double selection~\cite{fujii2009entanglement}.

\section{Distillation circuits}
\label{sec:distcircuits}
In this section we are concerned with finding circuits that achieve the highest fidelity for $n = 4$ to $8$ for an $n$-fold tensor product of a Werner state\footnote{We are not interested in the cases $n=2$ and $n=3$, since for those cases the concatenated DEJMPS protocols are optimal.}.

We first note that one could use techniques for general Clifford circuit decompositions to decompose the symplectic matrices of the form in \ref{lemma:nfoldwernerreduction}. However, we found that this would in general lead to circuits with high depths. 
Instead, we first find that any distillation protocol has a circuit in a given form. Then, we randomly generate circuits of that form, until we find circuits that achieve the highest fidelity, and have small depth.

\subsection{Reducing the circuit search space}
We use the Bruhat decomposition from~\cite{Bravyi2020, Maslov2018}, which allows to write any Clifford circuit $C$ in the form $C = FWF'$, with $F$ and, $F'$ elements of the so-called Borel subgroup\footnote{We use a different convention from~\cite{Bravyi2020, Maslov2018}, where the target index for the $\CNOT$ gates in the Borel subgroup is larger than the control index.} and $W$ a layer of Hadamard gates followed by a permutation $\sigma \in \mathcal{S}_n$. The Borel subgroup $\mathcal{B}_n$ is generated by $X_i$, $\Sgate_i$ for $1 \leq i \leq n$ and $\CNOT_{ij}$ with $1\leq j < i \leq n$. For convenience, we denote such $\CNOT$ gates as $\CNOT^{\uparrow}$ gates. Now note that the Borel subgroup  $\mathcal{B}_n$ is a subgroup of the distillation subgroup $\mathcal{D}_n$. This implies that the $F$ part of any circuit in the form $C = FWF'$ does not change the distillation statistics. Thus, any distillation protocol has a corresponding circuit of the form $WF'$. Furthermore, since the distillation subgroup $\mathcal{D}_n$ contains elements that arbitrarily permute qubits $2$ to $n$, we can restrict to permutations that are either the identity, or exchange qubit $1$ with $j$. In practice, we have found that it is sufficient to only consider $W = \Hgate_2\Hgate_3\ldots H_n$.

By the results from~\cite{Bravyi2020}, any element $F'$ from $\mathcal{B}_n$ can be written as a layer of $\CNOT^\uparrow$ gates, a layer of $\CZ$ gates, a layer of phase gates and a layer of Pauli gates. Firstly, the Pauli gates are in the kernel of $\phi$, and thus can be left out. Secondly, the layer of $S$ gates can be moved to the beginning. To see this, first note that phase gates commute with $\CZ$ gates. Then, since $\CNOT_{ij}\Sgate_i = \Sgate_i\CNOT_{ij}$ and $\Sgate_j\CNOT_{ij} = \CNOT_{ij}\CZ_{ij}\Sgate_i\Sgate_j$, the layer of $\Sgate$ gates can be moved to the beginning.
Since Werner states are invariant under $\Sgate$, the layer of $\Sgate$ gates can be removed without changing the distillation statistics. 
In the above process of moving the $\Sgate$ gates to the beginning, the layer that only had $\CNOT^\uparrow$ gates will now have $\CZ$ gates as well. In the binary picture we have the following identities,
\begin{gather}
\CNOT_{ij}\CZ_{kl} = \CZ_{kl}\CNOT_{ij}\ ,\\
\CNOT_{ij}\CZ_{ij} = \CZ_{ij}\CNOT_{ij}\ ,\\ \CNOT_{ik}\CZ_{ij} = \CZ_{ij}\CNOT_{ik}\ ,\\
\CNOT_{ik}\CZ_{jk} = \CZ_{ij}\CZ_{jk}\CNOT_{ik}\ ,
\end{gather}
where the $i, j, k, l$ are assumed to be distinct, and can be verified using Tables \ref{tab:rowoperations} and \ref{tab:columnoperations}. By using the above identities, the $\CZ$ gates can be moved through to the original layer of $\CZ$ gates.

It is thus sufficient to consider only elements $F'$ that consist of a layer of $\CNOT^\uparrow$ gates and a layer of $\CZ$ gates. Now, to find circuits we randomly generate circuits consisting of a layer of $\CNOT^\uparrow$ gates, a layer of $\CZ$ gates, and a Hadamard gate on all qubits except the first. We found several circuits that achieved the largest fidelity, and choose the one with smallest depth. We report the found circuits in Fig.~\ref{fig:opt_prots}.

\begin{figure}
\begin{subfigure}{.33\textwidth}
% \centering 
\[
  \Qcircuit @C=.7em @R=.7em {
         	&\qw &\targ 	&\qw 	        &\ctrl{1}  		&\qw    	    &\qw  	&\qw    &\qw   	    &\qw			\\ 
        	&\qw &\qw       &\ctrl{1}       &\control\qw    &\qw            &\qw  	&\qw	& \gate{H} & \measureD{Z}  	\\ 
        	&\qw &\qw  	    &\control\qw    &\qw          	&\ctrl{1}	    &\qw  	&\qw 	& \gate{H} & \measureD{Z}	\\ 
        	&\qw &\ctrl{-3} &\qw    	    &\qw		    &\control\qw    &\qw  	&\qw  	& \gate{H} & \measureD{Z}		
} 
\]
\caption{$n=4$, $\#$2-qubit gates $= 4$, depth $=3$.}
\label{fig:opt_prot_n=4}
\end{subfigure}%
\begin{subfigure}{.33\textwidth}
% \centering 
\[
  \Qcircuit @C=.7em @R=.7em {
        & 	&\qw 	&\targ	 	&\targ		&\qw		&\ctrl{2}	    &\qw		    &\qw		    &\qw	        &\qw	&\qw	&\qw        &\qw			\\ 
        &	&\qw  	&\qw		&\qw		&\qw		&\qw		    &\ctrl{3}	    &\ctrl{1}	    &\qw	        &\qw	&\qw	& \gate{H}  &\measureD{Z}  	\\ 
        &	&\qw  	&\ctrl{-2}  &\qw  	    &\targ		&\control\qw	&\qw	        &\control\qw	&\qw	        &\qw	&\qw 	& \gate{H}  &\measureD{Z}	\\ 
        &	&\qw  	&\qw	    &\qw  	    &\ctrl{-1}	&\qw		    &\qw		    &\qw    		&\ctrl{1}       &\qw	&\qw	& \gate{H}  &\measureD{Z}	\\ 
        &	&\qw  	&\qw 		&\ctrl{-4} 	&\qw	    &\qw		    &\control\qw	&\qw        	&\control\qw    &\qw	&\qw	& \gate{H}  &\measureD{Z}		
} 
\]
\caption{$n=5$, $\#$2-qubit gates $=7$, depth $=5$.}
\label{fig:opt_prot_n=5}
\end{subfigure}
\begin{subfigure}{0.33\textwidth}
\centering 
\[
  \Qcircuit @C=.7em @R=.7em {
        & 	&\qw 	&\targ	 	&\qw		&\targ		&\qw		&\qw		    &\qw		    &\ctrl{2}	    &\qw		    &\qw &\qw   &\qw	    &\qw			\\ 
        &	&\qw  	&\qw		&\targ		&\qw		&\qw		&\qw		    &\ctrl{1}	    &\qw		    &\ctrl{3}	    &\qw &\qw	& \gate{H}  &\measureD{Z}  	\\ 
        &	&\qw  	&\ctrl{-2}	&\ctrl{-1}  &\qw		&\targ		&\qw		    &\control\qw	&\control\qw	&\qw		    &\qw &\qw   & \gate{H}  &\measureD{Z}	\\ 
        &	&\qw  	&\qw		&\qw	  	&\qw		&\ctrl{-1}	&\qw		    &\qw		    &\qw		    &\qw		    &\qw &\qw   & \gate{H}  &\measureD{Z}	\\ 
        &	&\qw  	&\qw 		&\qw 		&\ctrl{-4}	&\qw		&\ctrl{1}	    &\qw		    &\qw		    &\control\qw	&\qw &\qw   & \gate{H}  &\measureD{Z}	\\	
        &	&\qw  	&\qw		&\qw	  	&\qw		&\qw		&\control\qw	&\qw		    &\qw		    &\qw		    &\qw &\qw   & \gate{H}  &\measureD{Z}	 
} 
\]
\caption{$n=6$, $\#$2-qubit gates $ = 8$, depth $=6$}
\label{fig:opt_prot_n=6}
\end{subfigure}

\begin{subfigure}{0.45\textwidth}
\centering 
\[
  \Qcircuit @C=.7em @R=.7em {
        & 	&\qw 	&\qw	 	&\targ		&\qw		&\targ		&\qw		    &\qw		    &\ctrl{2}	    &\qw		    &\qw		    &\qw		    &\qw		    &\qw &\qw   &\qw	    &\qw			\\ 
        &	&\qw  	&\qw		&\qw		&\qw		&\ctrl{-1}	&\qw		    &\ctrl{4}	    &\qw		    &\ctrl{2}	    &\qw		    &\qw		    &\qw		    &\qw &\qw	&\gate{H}  &\measureD{Z}  	\\ 
        &	&\qw  	&\qw  		&\ctrl{-2}  &\targ		&\qw		&\qw		    &\qw		    &\control\qw	&\qw		    &\ctrl{4}	    &\ctrl{1}	    &\qw		    &\qw &\qw   &\gate{H}  &\measureD{Z}	\\ 
        &	&\qw  	&\targ		&\qw	  	&\qw		&\qw		&\qw		    &\qw		    &\qw		    &\control\qw	&\qw		    &\control\qw	&\qw		    &\qw &\qw   &\gate{H}  &\measureD{Z}	\\ 
        &	&\qw  	&\ctrl{-1}	&\qw	  	&\ctrl{-2}	&\qw		&\qw		    &\qw		    &\qw		    &\qw		    &\qw		    &\qw		    &\ctrl{1}	    &\qw &\qw   &\gate{H}  &\measureD{Z}	\\ 
        &	&\qw  	&\qw		&\qw	  	&\qw		&\qw		&\ctrl{1}	    &\control\qw	&\qw		    &\qw		    &\qw		    &\qw		    &\control\qw	&\qw &\qw   &\gate{H}  &\measureD{Z}	\\ 
        &	&\qw  	&\qw 		&\qw 		&\qw		&\qw		&\control\qw	&\qw		    &\qw		    &\qw		    &\control\qw	&\qw		    &\qw		    &\qw &\qw   &\gate{H}  &\measureD{Z}		
} 
\]
\caption{$n=7$, $\#$2-qubit gates $=11$, depth $=6$.}
\label{fig:opt_prot_n=7}
\end{subfigure}%
\begin{subfigure}{0.45\textwidth}
\centering 
\[
  \Qcircuit @C=.7em @R=.7em {
        & 	&\qw 	&\qw	 	&\qw	  	&\qw		&\qw		&\targ		&\qw		&\qw		&\qw		    &\ctrl{6}	    &\qw		    &\qw		    &\qw		    &\qw		    &\qw &\qw   &\qw	    &\qw			\\ 
        &	&\qw  	&\qw		&\qw	  	&\targ		&\qw		&\qw		&\qw		&\qw		&\qw		    &\qw		    &\qw		    &\qw		    &\qw		    &\ctrl{2}	    &\qw &\qw	&\gate{H}   &\measureD{Z}  	\\ 
        &	&\qw  	&\targ		&\qw	  	&\ctrl{-1}  &\qw		&\qw		&\qw		&\qw		&\ctrl{2}	    &\qw		    &\qw    	    &\qw	        &\ctrl{4}		&\qw		    &\qw &\qw   &\gate{H}   &\measureD{Z}	\\ 
        &	&\qw  	&\qw		&\qw	  	&\qw	  	&\targ		&\ctrl{-3}	&\qw		&\targ		&\qw		    &\qw		    &\qw		    &\ctrl{3}		&\qw	        &\control\qw	&\qw &\qw   &\gate{H}   &\measureD{Z}	\\ 
        &	&\qw  	&\qw		&\qw	  	&\qw	  	&\qw		&\qw		&\qw		&\qw		&\control\qw    &\qw		    &\ctrl{1}   	&\qw		    &\qw		    &\qw		    &\qw &\qw   &\gate{H}   &\measureD{Z}	\\ 
        &	&\qw  	&\qw		&\targ  	&\qw	  	&\qw		&\qw		&\qw		&\ctrl{-2}	&\qw	        &\qw		    &\control\qw    &\qw		    &\qw		    &\qw		    &\qw &\qw   &\gate{H}   &\measureD{Z}	\\ 
        &	&\qw  	&\qw    	&\ctrl{-1} 	&\qw	  	&\qw		&\qw		&\targ		&\qw		&\qw		    &\control\qw    &\qw		    &\control\qw	&\control\qw	&\qw		    &\qw &\qw   &\gate{H}   &\measureD{Z}	\\ 
        &	&\qw  	&\ctrl{-5}	&\qw	  	&\qw 		&\ctrl{-4}	&\qw		&\ctrl{-1}	&\qw		&\qw		    &\qw		    &\qw		    &\qw		    &\qw		    &\qw		    &\qw &\qw   &\gate{H}   &\measureD{Z}		
} 
\]
\caption{$n=8$, $\#$2-qubit gates $=13$, depth $=7$.}
\label{fig:opt_prot_n=8}
\end{subfigure}
\caption{Circuits that achieve the maximum fidelity for $n$. These circuits are applied by both Alice and Bob, after which they measure the last $n-1$ qubits, and communicate their outcomes to each other. When the outcomes for all individual qubit pairs are correlated, the distillation protocol was deemed successful.}
\label{fig:opt_prots}
\end{figure}
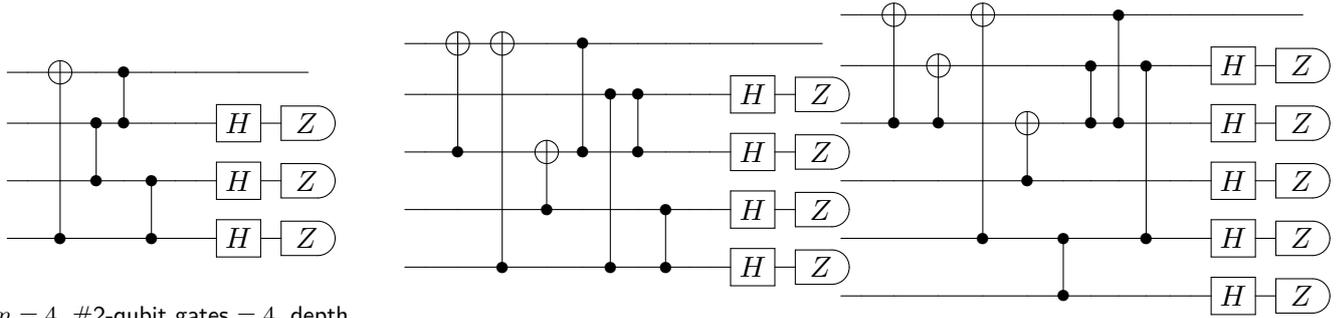
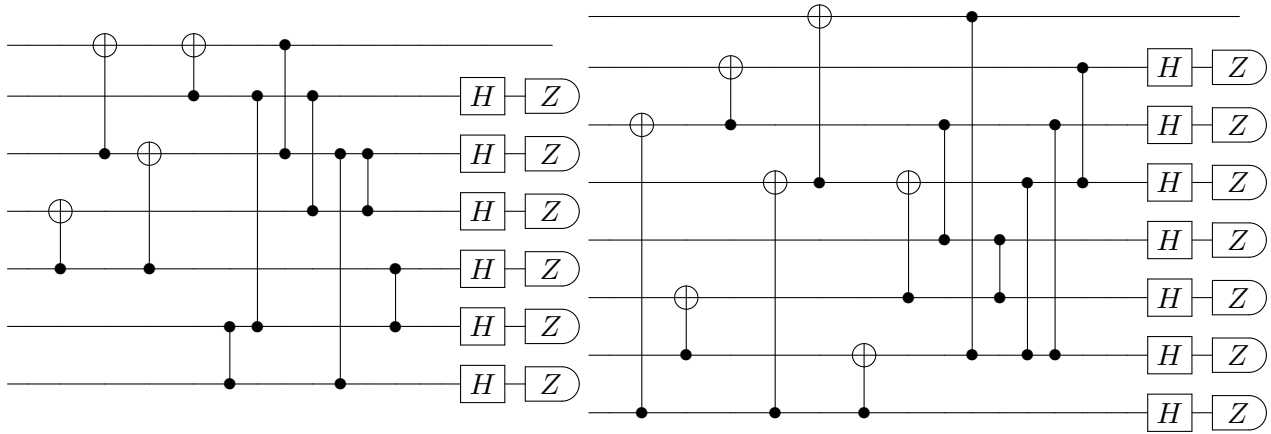

\FloatBarrier
\section{Analytical expressions}\label{sec:diststatistics}
We report here the analytical expressions of the fidelity and success probability that correspond to the found optimal schemes. The input state is an $n$-fold tensor product of a Werner state with fidelity $F$. For completeness, we report here as well the distillation statistics expressed in the infidelity $\epsilon \equiv 1-F$, and the scaling of the output fidelity as a function of the infidelity.
\vspace{5mm}

\begin{table}[h]
    \centering
\scalebox{1.1}{
	\begin{tabular}{c || c |}
		$n$ & $p_\textrm{suc}$\\
		&\\[-0.8em]
		\hline
		&\\[-0.8em]
 $2$ & $\frac{8}{9}F^2-\frac{4}{9}F+\frac{5}{9}$\\
 &\\[-0.8em]
		\hline
		&\\[-0.8em]
		$3$ & $\frac{32}{27}F^3-\frac{4}{9}F^2+\frac{7}{27}$\\
		&\\[-0.8em]
		\hline
		&\\[-0.8em]
 $4$ & $\frac{32}{27}F^4-\frac{4}{9}F^2+\frac{4}{27}F+\frac{1}{9}$\\
 &\\[-0.8em]
		\hline
		&\\[-0.8em]
		$5$ & $ \frac{80}{27}F^4 - \frac{80}{27}F^3 + \frac{10}{9}F^2 - \frac{5}{27}F + \frac{2}{27}$\\
			&\\[-0.8em]
		\hline
		&\\[-0.8em]
				$6$ & $  \frac{128}{243}F^6 + \frac{320}{243}F^5 - \frac{256}{243}F^4 + \frac{16}{243}F^3 + \frac{40}{243}F^2 - \frac{14}{243}F + \frac{1}{27}$\\
		&\\[-0.8em]
		\hline
		&\\[-0.8em]
						$7$ & $   \frac{2048}{2187}F^7 - \frac{128}{2187}F^6 + \frac{320}{729}F^5 - \frac{796}{2187}F^4 - \frac{44}{2187}F^3 + \frac{49}{729}F^2 - \frac{37}{2187}F + \frac{37}{2187}$\\
		&\\[-0.8em]
		\hline
		&\\[-0.8em]
						$8$ & $    \frac{6656}{6561}F^8 - \frac{1024}{6561}F^7 + \frac{1664}{6561}F^6 - \frac{64}{6561}F^5 - \frac{1120}{6561}F^4 + \frac{416}{6561}F^3 - \frac{4}{6561}F^2 - \frac{16}{6561}F + \frac{53}{6561}$\\
			&\\[-1.3em]
	\end{tabular}}
	\vspace{3.5mm}
\caption{Success probability for the protocols with the highest output fidelity for $n=2$ to $8$.}
	\label{table:distillationstatistics}
\end{table}

\begin{table}[h]
    \centering
\scalebox{1.1}{
	\begin{tabular}{c || c | c }
		$n$  & $p_\textrm{suc}\cdot F_\textrm{out}$ \\
		&\\[-0.8em]
		\hline
		&\\[-0.8em]
 $2$ & $\frac{10}{9}F^2-\frac{2}{9}F+\frac{1}{9}$\\
 &\\[-0.8em]
		\hline
		&\\[-0.8em]
		$3$ & $\frac{28}{27}F^3-\frac{1}{9}F+\frac{2}{27}$\\
		&\\[-0.8em]
		\hline
		&\\[-0.8em]
 $4$ & $\frac{8}{9}F^4+\frac{8}{27}F^3-\frac{2}{9}F^2+\frac{1}{27}$\\
 &\\[-0.8em]
		\hline
		&\\[-0.8em]
		$5$ & $ \frac{32}{27}F^5 - \frac{20}{27}F^4 + \frac{10}{9}F^3 - \frac{20}{27}F^2 + \frac{5}{27}F$ \\
			&\\[-0.8em]
		\hline
		&\\[-0.8em]
				$6$ & $ \frac{32}{27}F^6 - \frac{112}{243}F^5 + \frac{80}{243}F^4 + \frac{8}{243}F^3 - \frac{32}{243}F^2 + \frac{10}{243}F + \frac{1}{243}$\\
		&\\[-0.8em]
		\hline
		&\\[-0.8em]
						$7$ & $ \frac{2368}{2187}F^7 - \frac{592}{2187}F^6 + \frac{196}{729}F^5 - \frac{44}{2187}F^4 - \frac{199}{2187}F^3 + \frac{20}{729}F^2 - \frac{2}{2187}F + \frac{8}{2187}$\\
		&\\[-0.8em]
		\hline
		&\\[-0.8em]
						$8$ & $  \frac{6784}{6561}F^8 - \frac{512}{6561}F^7 - \frac{32}{6561}F^6 + \frac{832}{6561}F^5 - \frac{560}{6561}F^4 - \frac{8}{6561}F^3 + \frac{52}{6561}F^2 - \frac{8}{6561}F + \frac{13}{6561}$\\
			&\\[-1.3em]
	\end{tabular}}
	\vspace{3.5mm}
	\caption{Product of the success probability and the output fidelity for the protocols with the highest output fidelity for $n=2$ to $8$.}
	\label{table:distillationstatistics2}
\end{table}

\begin{table}[h]
    \centering
\scalebox{1.27}{
	\begin{tabular}{c || c |}
		$n$ & $p_\textrm{suc}$\\
		&\\[-0.8em]
		\hline
		&\\[-0.8em]
 $2$ & $1-\frac{4}{3}\epsilon+\frac{8}{9}\epsilon^2$\\
 &\\[-0.8em]
		\hline
		&\\[-0.8em]
		$3$ & $1-2\epsilon+\frac{4}{3}\epsilon^2$\\
		&\\[-0.8em]
		\hline
		&\\[-0.8em]
 $4$ & $1-2\epsilon+\frac{4}{3}\epsilon^2-\frac{8}{27}\epsilon^3$\\
 &\\[-0.8em]
		\hline
		&\\[-0.8em]
		$5$ & $ 1-\frac{14}{3}\epsilon + \frac{28}{3}\epsilon^2 - \frac{256}{27}\epsilon^3 + \frac{400}{81}\epsilon^4 - \frac{256}{243}\epsilon^5$\\
			&\\[-0.8em]
		\hline
		&\\[-0.8em]
				$6$ & $  1-5\epsilon + \frac{32}{3}\epsilon^2 - 12 \epsilon^3 + \frac{608}{81}\epsilon^4 - \frac{608}{243}\epsilon^5 + \frac{256}{729}\epsilon^6$\\
		&\\[-0.8em]
		\hline
		&\\[-0.8em]
						$7$ & $   1-7\epsilon +\frac{190}{9}\epsilon^2 - \frac{944}{27}\epsilon^3 + \frac{928}{27}\epsilon^4 - \frac{544}{27}\epsilon^5 + \frac{1600}{243}\epsilon^6 - \frac{2048}{2187}\epsilon^7$\\
		&\\[-0.8em]
		\hline
		&\\[-0.8em]
						$8$ & $    1-\frac{23}{3}\epsilon +\frac{244}{9}\epsilon^2 - \frac{1540}{27}\epsilon^3 + \frac{6280}{81}\epsilon^4 - \frac{16832}{243}\epsilon^5 + \frac{28768}{729}\epsilon^6 - \frac{9472}{729}\epsilon^7 + \frac{4096}{2187}\epsilon^8$\\
			&\\[-1.3em]
	\end{tabular}}
	\vspace{3.5mm}
\caption{Success probability for the protocols with the highest output fidelity for $n=2$ to $8$, expressed in the infidelity $\epsilon\equiv 1-F$.}
	\label{table:distillationstatisticsinfidelity}
\end{table}

\begin{table}[h]
    \centering
\scalebox{1.27}{
	\begin{tabular}{c || c | c }
		$n$  & $p_\textrm{suc}\cdot F_\textrm{out}$ \\
		&\\[-0.8em]
		\hline
		&\\[-0.8em]
 $2$ & $1-2\epsilon + \frac{10}{9}\epsilon^2$\\
 &\\[-0.8em]
		\hline
		&\\[-0.8em]
		$3$ & $1-3\epsilon+\frac{10}{3}\epsilon^2-\frac{4}{3}\epsilon^3$\\
		&\\[-0.8em]
		\hline
		&\\[-0.8em]
 $4$ & $1-3\epsilon+\frac{10}{3}\epsilon^2-\frac{44}{27}\epsilon^3+\frac{8}{27}^4$\\
 &\\[-0.8em]
		\hline
		&\\[-0.8em]
		$5$ & $ 1-5\epsilon +\frac{92}{9}\epsilon^2-\frac{284}{27}\epsilon^3 + \frac{440}{81}\epsilon^4-\frac{272}{243}\epsilon^5$ \\
			&\\[-0.8em]
		\hline
		&\\[-0.8em]
				$6$ & $ 1-\frac{17}{3}\epsilon+\frac{122}{9}\epsilon^2-\frac{466}{27}\epsilon^3+\frac{992}{81}\epsilon^4-\frac{1112}{243}\epsilon^5+\frac{512}{729}\epsilon^6$\\
		&\\[-0.8em]
		\hline
		&\\[-0.8em]
						$7$ & $ 1-7\epsilon+\frac{190}{9}\epsilon^2-\frac{320}{9}\epsilon^3+\frac{2936}{81}\epsilon^4-\frac{1816}{81}\epsilon^5+\frac{5680}{729}\epsilon^6-\frac{2560}{2187}\epsilon^7$\\
		&\\[-0.8em]
		\hline
		&\\[-0.8em]
						$8$ & $  1-8\epsilon + \frac{259}{9}\epsilon^2 - \frac{544}{9}\epsilon^3 + \frac{2180}{27}\epsilon^4 - \frac{17000}{243}\epsilon^5 + \frac{27872}{729}\epsilon^6 - \frac{2912}{243}\epsilon^7 + \frac{3584}{2187}\epsilon^8$\\
			&\\[-1.3em]
	\end{tabular}}
	\vspace{3.5mm}
	\caption{Product of the success probability and the output fidelity for the protocols with the highest output fidelity for $n=2$ to $8$, expressed in the infidelity $\epsilon\equiv 1-F$.}
	\label{table:distillationstatistics2infidelity}
\end{table}

\begin{table}[h]
    \centering
\scalebox{1.27}{
	\begin{tabular}{c || c | c }
		$n$  & $F_\textrm{out}$ \\
		&\\[-0.8em]
		\hline
		&\\[-0.8em]
 $2$ & $1-\frac{2}{3}\epsilon -\mathcal{O}\left(\epsilon^2\right)$\\
 &\\[-0.8em]
		\hline
		&\\[-0.8em]
		$3$ & $1-\frac{1}{3}\epsilon -\mathcal{O}\left(\epsilon^2\right)$\\
		&\\[-0.8em]
		\hline
		&\\[-0.8em]
 $4$ & $1-\frac{2}{3}\epsilon^2-\mathcal{O}\left(\epsilon^3\right)$\\
 &\\[-0.8em]
		\hline
		&\\[-0.8em]
		$5$ & $ 1-\frac{10}{9}\epsilon^3 -\mathcal{O}\left(\epsilon^4\right)$ \\
			&\\[-0.8em]
		\hline
		&\\[-0.8em]
				$6$ & $ 1-\frac{8}{9}\epsilon^3 -\mathcal{O}\left(\epsilon^4\right)$\\
		&\\[-0.8em]
		\hline
		&\\[-0.8em]
						$7$ & $ 1-\frac{13}{27}\epsilon^3 -\mathcal{O}\left(\epsilon^4\right)$\\
		&\\[-0.8em]
		\hline
		&\\[-0.8em]
						$8$ & $ 1-\frac{8}{27}\epsilon^3 -\mathcal{O}\left(\epsilon^4\right)$\\
			&\\[-1.3em]
	\end{tabular}}
	\vspace{3.5mm}
	\caption{Scaling of the output fidelity around $\epsilon\approx0$ for the protocols with the highest output fidelity for $n=2$ to $8$, where $\epsilon\equiv 1-F$.}
	\label{table:distillationstatisticsinfidelityseries}
\end{table}

\end{document}